\documentclass[lettersize,journal]{IEEEtran}
\usepackage{array}
\usepackage{textcomp}
\usepackage{stfloats}
\usepackage{url}
\usepackage{verbatim}
\usepackage{flushend}
\usepackage{graphicx}
\usepackage{float}
\usepackage{cite}
\usepackage[format=hang,singlelinecheck=false]{subcaption}
\usepackage{caption}

\usepackage[cmex10]{amsmath}
\usepackage{mathtools}
\usepackage{amsfonts}
\usepackage{amssymb}
\usepackage{amsthm}

\newtheorem{proposition}{Proposition}
\newtheorem{remark}{Remark}

\newtheorem{theorem}{Theorem}
\newtheorem{lemma}{Lemma}
\newtheorem{definition}{Definition}
\newtheorem{corollary}{Corollary}
\newtheorem{conjecture}{Conjecture}
\DeclareMathAlphabet{\mathppl}{T1}{ppl}{m}{it}
\DeclareMathAlphabet{\mathphv}{T1}{phv}{m}{it}
\DeclareMathAlphabet{\mathpzc}{T1}{pzc}{m}{it}
\newlength{\norlen} 

\DeclareMathOperator{\sign}{sgn}

\newcommand{\Vt}[1]{\mathbf{\lowercase{#1}}}

\newcommand{\vtA}{\Vt{A}}
\newcommand{\vtB}{\Vt{B}}

\newcommand{\vtE}{\Vt{E}}
\newcommand{\vtF}{\Vt{F}}

\newcommand{\vtM}{\Vt{M}}
\newcommand{\vtN}{\Vt{N}}

\newcommand{\vtS}{\Vt{S}}

\newcommand{\vtW}{\Vt{W}}
\newcommand{\vtX}{\Vt{X}}
\newcommand{\vtY}{\Vt{Y}}

\newcommand{\vtAlpha}{\Vt{\boldsymbol{\alpha}}}

\newcommand{\vtZero}{\Vt{0}}
\usepackage{algorithm}
\usepackage{algorithmicx}
\usepackage{algpseudocode}
\usepackage{multicol}
\algblock[<block>]{ps}{endps}
\algrenewtext{ps}{\textbf{Parameter Server:}\ }
\algtext*{endps}
\algblock[<block>]{init}{endinit}
\algrenewtext{init}{\textbf{Initialize:}\ }
\algtext*{endinit}
\algblock[<block>]{device}{enddevice}
\algrenewtext{device}{\textbf{Device $i$:}\ }
\algtext*{enddevice}
\algblock[<block>]{devices}{enddevices}
\algrenewtext{devices}{\textbf{all devices simultaneously:}\ }
\algtext*{enddevices}

\usepackage{tikz}
\usepackage{subfig}
\usepackage{pgfplots}
\pgfplotsset{compat=1.18}

\usepackage{xcolor}
\definecolor{antiquewhite}{rgb}{0.98, 0.92, 0.84}
\definecolor{antiquefuchsia}{rgb}{0.57, 0.36, 0.51}\definecolor{chestnut}{rgb}{0.8, 0.36, 0.36}
\definecolor{airforceblue}{rgb}{0.36, 0.54, 0.66}
\definecolor{cadmiumorange}{rgb}{0.93, 0.53, 0.18}
\definecolor{bleudefrance}{rgb}{0.19, 0.55, 0.91}
\definecolor{carolinablue}{rgb}{0.6, 0.73, 0.89}
\definecolor{blue(ncs)}{rgb}{0.0, 0.53, 0.74}
\definecolor{dodgerblue}{rgb}{0.12, 0.56, 1.0}
\definecolor{cssgreen}{rgb}{0.0, 0.5, 0.0}
\definecolor{cadmiumgreen}{rgb}{0.0, 0.42, 0.24}
\definecolor{cadmiumorange}{rgb}{0.93, 0.53, 0.18}
\definecolor{amaranth}{rgb}{0.9, 0.17, 0.31}
\definecolor{bluegray}{rgb}{0.4, 0.6, 0.8}
\definecolor{cadmiumgreen}{rgb}{0.0, 0.42, 0.24}
\definecolor{carnationpink}{rgb}{1.0, 0.65, 0.79}
\definecolor{copper}{rgb}{0.72, 0.45, 0.2}
\definecolor{dandelion}{rgb}{0.94, 0.88, 0.19}

\usepackage[nolist,withpage]{acronym}
\usepackage{etoolbox}
\makeatletter
\newif\if@in@acrolist
\AtBeginEnvironment{acronym}{\@in@acrolisttrue}
\newrobustcmd{\LU}[2]{\if@in@acrolist#1\else#2\fi}

\newcommand{\ACF}[1]{{\@in@acrolisttrue\acf{#1}}}
\makeatother

\begin{document}

\begin{acronym}[LTE-Advanced]
  \acro{1bCS}{1-bit \LU{C}{c}ompressed \LU{S}{s}ensing}
  \acro{2G}{Second Generation}
  \acro{3-DAP}{3-Dimensional Assignment Problem}
  \acro{3G}{3$^\text{rd}$~Generation}
  \acro{3GPP}{3$^\text{rd}$~Generation Partnership Project}
  \acro{4G}{4$^\text{th}$~Generation}
  \acro{5G}{5$^\text{th}$~Generation}
  \acro{AA}{Antenna Array}
  \acro{AC}{Admission Control}
  \acro{AD}{Attack-Decay}
  \acro{ADC}{analog-to-digital converter}
  \acro{ADMM}{alternating direction method of multipliers}
  \acro{ADSL}{Asymmetric Digital Subscriber Line}
  \acro{AHW}{Alternate Hop-and-Wait}
  \acro{AirComp}{over-the-air computation}
  \acro{AirGT}{\LU{O}{o}ver-the-\LU{A}{a}ir \LU{G}{g}roup \LU{T}{t}esting}
  \acro{AMC}{Adaptive Modulation and Coding}
  \acro{AP}{\LU{A}{a}ccess \LU{P}{p}oint}
  \acro{APA}{Adaptive Power Allocation}
  \acro{ARMA}{Autoregressive Moving Average}
  \acro{ARQ}{\LU{A}{a}utomatic \LU{R}{r}epeat \LU{R}{r}equest}
  \acro{ATES}{Adaptive Throughput-based Efficiency-Satisfaction Trade-Off}
  \acro{AWGN}{additive white Gaussian noise}
  \acro{BAA}{\LU{B}{b}roadband \LU{A}{a}nalog \LU{A}{a}ggregation}
  \acro{BB}{Branch and Bound}
  \acro{BCD}{block coordinate descent}
  \acro{BD}{Block Diagonalization}
  \acro{BER}{Bit Error Rate}
  \acro{BF}{Best Fit}
  \acro{BFD}{bidirectional full duplex}
  \acro{BLER}{BLock Error Rate}
  \acro{BPC}{Binary Power Control}
  \acro{BPSK}{Binary Phase-Shift Keying}
  \acro{BRA}{Balanced Random Allocation}
  \acro{BS}{base station}
  \acro{BSUM}{block successive upper-bound minimization}
  \acro{CAP}{Combinatorial Allocation Problem}
  \acro{CAPEX}{Capital Expenditure}
  \acro{CBF}{Coordinated Beamforming}
  \acro{CBR}{Constant Bit Rate}
  \acro{CBS}{Class Based Scheduling}
  \acro{CC}{Congestion Control}
  \acro{CDF}{Cumulative Distribution Function}
  \acro{CDMA}{Code-Division Multiple Access}
  \acro{CE}{\LU{C}{c}hannel \LU{E}{e}stimation}
  \acro{CL}{Closed Loop}
  \acro{CLPC}{Closed Loop Power Control}
  \acro{CML}{centralized machine learning}
  \acro{CNR}{Channel-to-Noise Ratio}
  \acro{CNN}{\LU{C}{c}onvolutional \LU{N}{n}eural \LU{N}{n}etwork}
  \acro{COMP}{combinatorial orthogonal matching pursuit}
  \acro{CPA}{Cellular Protection Algorithm}
  \acro{CPICH}{Common Pilot Channel}
  \acro{CoCoA}{\LU{C}{c}ommunication efficient distributed dual \LU{C}{c}oordinate \LU{A}{a}scent}
  \acro{CoMAC}{\LU{C}{c}omputation over \LU{M}{m}ultiple-\LU{A}{a}ccess \LU{C}{c}hannels}
  \acro{CoMP}{Coordinated Multi-Point}
  \acro{CQI}{Channel Quality Indicator}
  \acro{CRC}{cyclic redundancy check}
  \acro{CRM}{Constrained Rate Maximization}
	\acro{CRN}{Cognitive Radio Network}
  \acro{CS}{\LU{C}{c}ompressed \LU{S}{s}ensing}
  \acro{CSI}{\LU{C}{c}hannel \LU{S}{s}tate \LU{I}{i}nformation}
  \acro{CSMA}{\LU{C}{c}arrier \LU{S}{s}ense \LU{M}{m}ultiple \LU{A}{a}ccess}
  \acro{CUE}{Cellular User Equipment}
  \acro{D2D}{device-to-device}
  \acro{DAC}{digital-to-analog converter}
  \acro{DC}{direct current}
  \acro{DCA}{Dynamic Channel Allocation}
  \acro{DE}{Differential Evolution}
  \acro{DFT}{Discrete Fourier Transform}
  \acro{DIST}{Distance}
  \acro{DL}{downlink}
  \acro{DMA}{Double Moving Average}
  \acro{DML}{Distributed Machine Learning}
  \acro{DMRS}{demodulation reference signal}
  \acro{D2DM}{D2D Mode}
  \acro{DMS}{D2D Mode Selection}
  \acro{DPC}{Dirty Paper Coding}
  \acro{DRA}{Dynamic Resource Assignment}
  \acro{DSA}{Dynamic Spectrum Access}
  \acro{DSGD}{\LU{D}{d}istributed \LU{S}{s}tochastic \LU{G}{g}radient \LU{D}{d}escent}
  \acro{DSM}{Delay-based Satisfaction Maximization}
  \acro{ECC}{Electronic Communications Committee}
  \acro{EFLC}{Error Feedback Based Load Control}
  \acro{EI}{Efficiency Indicator}
  \acro{eNB}{Evolved Node B}
  \acro{EPA}{Equal Power Allocation}
  \acro{EPC}{Evolved Packet Core}
  \acro{EPS}{Evolved Packet System}
  \acro{E-UTRAN}{Evolved Universal Terrestrial Radio Access Network}
  \acro{ES}{Exhaustive Search}
  \acro{FA}{\LU{F}{f}ederated \LU{A}{a}nalytics}
  \acro{FC}{\LU{F}{f}usion \LU{C}{c}enter}
  \acro{FD}{\LU{F}{f}ederated \LU{D}{d}istillation}
  \acro{FDD}{frequency division duplex}
  \acro{FDM}{Frequency Division Multiplexing}
  \acro{FDMA}{\LU{F}{f}requency \LU{D}{d}ivision \LU{M}{m}ultiple \LU{A}{a}ccess}
  \acro{FedAvg}{\LU{F}{f}ederated \LU{A}{a}veraging}
  \acro{FER}{Frame Erasure Rate}
  \acro{FF}{Fast Fading}
  \acro{FL}{federated learning}
  \acro{FSB}{Fixed Switched Beamforming}
  \acro{FEC}{forward error correcting}
  \acro{FST}{Fixed SNR Target}
  \acro{FTP}{File Transfer Protocol}
  \acro{GA}{Genetic Algorithm}
  \acro{GBR}{Guaranteed Bit Rate}
  \acro{GD}{Gradient Descent}
  \acro{GLR}{Gain to Leakage Ratio}
  \acro{GOS}{Generated Orthogonal Sequence}
  \acro{GPL}{GNU General Public License}
  \acro{GRP}{Grouping}
  \acro{GT}{\LU{G}{g}roup \LU{T}{t}esting}
  \acro{HARQ}{Hybrid Automatic Repeat Request}
  \acro{HD}{half-duplex}
  \acro{HMS}{Harmonic Mode Selection}
  \acro{HOL}{Head Of Line}
  \acro{HSDPA}{High-Speed Downlink Packet Access}
  \acro{HSPA}{High Speed Packet Access}
  \acro{HTTP}{HyperText Transfer Protocol}
  \acro{ICMP}{Internet Control Message Protocol}
  \acro{ICI}{Intercell Interference}
  \acro{ID}{Identification}
  \acro{IETF}{Internet Engineering Task Force}
  \acro{iff}{if and only if}
  \acro{ILP}{Integer Linear Program}
  \acro{JRAPAP}{Joint RB Assignment and Power Allocation Problem}
  \acro{UID}{Unique Identification}
  \acro{IID}{\LU{I}{i}ndependent and \LU{I}{i}dentically \LU{D}{d}istributed}
  \acro{IIR}{Infinite Impulse Response}
  \acro{ILP}{Integer Linear Problem}
  \acro{IMT}{International Mobile Telecommunications}
  \acro{INV}{Inverted Norm-based Grouping}
  \acro{IoT}{Internet of Things}
  \acro{IP}{Integer Programming}
  \acro{IPv6}{Internet Protocol Version 6}
  \acro{ISD}{Inter-Site Distance}
  \acro{ISI}{Inter Symbol Interference}
  \acro{ITU}{International Telecommunication Union}
  \acro{IQ}{in-phase and quadrature}
  \acro{JAFM}{joint assignment and fairness maximization}
  \acro{JAFMA}{joint assignment and fairness maximization algorithm}
  \acro{JOAS}{Joint Opportunistic Assignment and Scheduling}
  \acro{JOS}{Joint Opportunistic Scheduling}
  \acro{JP}{Joint Processing}
	\acro{JS}{Jump-Stay}
  \acro{KKT}{Karush-Kuhn-Tucker}
  \acro{L3}{Layer-3}
  \acro{LAC}{Link Admission Control}
  \acro{LA}{Link Adaptation}
  \acro{LC}{Load Control}
  \acro{LDC}{\LU{L}{l}earning-\LU{D}{d}riven \LU{C}{c}ommunication}
  \acro{LDPC}{low-density parity-check}
  \acro{LO}{\LU{L}{l}ocal \LU{O}{o}scillator}
  \acro{LOS}{line of sight}
  \acro{LP}{Linear Programming}
  \acro{LTE}{Long Term Evolution}
	\acro{LTE-A}{\ac{LTE}-Advanced}
  \acro{LTE-Advanced}{Long Term Evolution Advanced}
  \acro{M2M}{Machine-to-Machine}
  \acro{MAC}{multiple access channel}
  \acro{MANET}{Mobile Ad hoc Network}
  \acro{MC}{Modular Clock}
  \acro{MCS}{modulation and coding scheme}
  \acro{MDB}{Measured Delay Based}
  \acro{MDI}{Minimum D2D Interference}
  \acro{MF}{Matched Filter}
  \acro{MG}{Maximum Gain}
  \acro{MGF}{Moment Generating Function}
  \acro{MH}{Multi-Hop}
  \acro{MIMO}{\LU{M}{m}ultiple \LU{I}{i}nput \LU{M}{m}ultiple \LU{O}{o}utput}
  \acro{MINLP}{mixed integer nonlinear programming}
  \acro{MIP}{Mixed Integer Programming}
  \acro{MISO}{multiple input single output}
  \acro{ML}{\LU{M}{m}achine \LU{L}{l}earning}
  \acro{MLWDF}{Modified Largest Weighted Delay First}
  \acro{MME}{Mobility Management Entity}
  \acro{MMSE}{minimum mean squared error}
  \acro{MOS}{Mean Opinion Score}
  \acro{MPF}{Multicarrier Proportional Fair}
  \acro{MRA}{Maximum Rate Allocation}
  \acro{MR}{Maximum Rate}
  \acro{MRC}{Maximum Ratio Combining}
  \acro{MRT}{Maximum Ratio Transmission}
  \acro{MRUS}{Maximum Rate with User Satisfaction}
  \acro{MS}{Mode Selection}
  \acro{MSE}{\LU{M}{m}ean \LU{S}{s}quared \LU{E}{e}rror}
  \acro{MSI}{Multi-Stream Interference}
  \acro{MTC}{Machine-Type Communication}
  \acro{MTSI}{Multimedia Telephony Services over IMS}
  \acro{MTSM}{Modified Throughput-based Satisfaction Maximization}
  \acro{MU-MIMO}{Multi-User Multiple Input Multiple Output}
  \acro{MU}{Multi-User}
  \acro{MV}{\LU{M}{m}ajority \LU{V}{v}ote}
  \acro{MVCS}{majority vote compressed sensing}
  \acro{NAS}{Non-Access Stratum}
  \acro{NB}{Node B}
	\acro{NCL}{Neighbor Cell List}
 \acro{NCOMP}{Noisy Combinatorial Orthogonal Matching Pursuit}
  \acro{NLP}{Nonlinear Programming}
  \acro{NLOS}{non-line of sight}
  \acro{NMSE}{Normalized Mean Square Error}
  \acro{NOMA}{\LU{N}{n}on-\LU{O}{o}rthogonal \LU{M}{m}ultiple \LU{A}{a}ccess}
  \acro{NORM}{Normalized Projection-based Grouping}
  \acro{NP}{non-polynomial time}
  \acro{NRT}{Non-Real Time}
  \acro{NSPS}{National Security and Public Safety Services}
  \acro{O2I}{Outdoor to Indoor}
  \acro{OFDMA}{\LU{O}{o}rthogonal \LU{F}{f}requency \LU{D}{d}ivision \LU{M}{m}ultiple \LU{A}{a}ccess}
  \acro{OFDM}{Orthogonal Frequency Division Multiplexing}
  \acro{OFPC}{Open Loop with Fractional Path Loss Compensation}
	\acro{O2I}{Outdoor-to-Indoor}
  \acro{OL}{Open Loop}
  \acro{OLPC}{Open-Loop Power Control}
  \acro{OL-PC}{Open-Loop Power Control}
  \acro{OPEX}{Operational Expenditure}
  \acro{ORB}{Orthogonal Random Beamforming}
  \acro{JO-PF}{Joint Opportunistic Proportional Fair}
  \acro{OSI}{Open Systems Interconnection}
  \acro{PAIR}{D2D Pair Gain-based Grouping}
  \acro{PAPR}{Peak-to-Average Power Ratio}
  \acro{P2P}{Peer-to-Peer}
  \acro{PC}{Power Control}
  \acro{PCI}{Physical Cell ID}
  \acro{PDCCH}{physical downlink control channel}
  \acro{PDD}{penalty dual decomposition}
  \acro{PDF}{Probability Density Function}
  \acro{PER}{Packet Error Rate}
  \acro{PF}{Proportional Fair}
  \acro{P-GW}{Packet Data Network Gateway}
  \acro{PL}{Pathloss}
  \acro{PRB}{Physical Resource Block}
  \acro{PROJ}{Projection-based Grouping}
  \acro{ProSe}{Proximity Services}
  \acro{PS}{\LU{P}{p}arameter \LU{S}{s}erver}
  \acro{PSO}{Particle Swarm Optimization}
  \acro{PUCCH}{physical uplink control channel}
  \acro{PZF}{Projected Zero-Forcing}
  \acro{QAM}{Quadrature Amplitude Modulation}
  \acro{QoS}{quality of service}
  \acro{QPSK}{Quadri-Phase Shift Keying}
  \acro{RAISES}{Reallocation-based Assignment for Improved Spectral Efficiency and Satisfaction}
  \acro{RAN}{Radio Access Network}
  \acro{RA}{Resource Allocation}
  \acro{RAT}{Radio Access Technology}
  \acro{RATE}{Rate-based}
  \acro{RB}{resource block}
  \acro{RBG}{Resource Block Group}
  \acro{REF}{Reference Grouping}
  \acro{RF}{radio frequency}
  \acro{RIP}{restricted isometry property}
  \acro{RLC}{Radio Link Control}
  \acro{RM}{Rate Maximization}
  \acro{RNC}{Radio Network Controller}
  \acro{RND}{Random Grouping}
  \acro{RRA}{Radio Resource Allocation}
  \acro{RRM}{\LU{R}{r}adio \LU{R}{r}esource \LU{M}{m}anagement}
  \acro{RSCP}{Received Signal Code Power}
  \acro{RSRP}{reference signal receive power}
  \acro{RSRQ}{Reference Signal Receive Quality}
  \acro{RR}{Round Robin}
  \acro{RRC}{Radio Resource Control}
  \acro{RSSI}{received signal strength indicator}
  \acro{RT}{Real Time}
  \acro{RU}{Resource Unit}
  \acro{RUNE}{RUdimentary Network Emulator}
  \acro{RV}{random variable}
  \acro{SAA}{Small Argument Approximation}
  \acro{SAC}{Session Admission Control}
  \acro{SCM}{Spatial Channel Model}
  \acro{SC-FDMA}{Single Carrier - Frequency Division Multiple Access}
  \acro{SD}{Soft Dropping}
  \acro{S-D}{Source-Destination}
  \acro{SDPC}{Soft Dropping Power Control}
  \acro{SDMA}{Space-Division Multiple Access}
  \acro{SDR}{semidefinite relaxation}
  \acro{SDP}{semidefinite programming}
  \acro{SecAgg}{\LU{S}{s}ecure \LU{A}{a}ggregation}
  \acro{SER}{Symbol Error Rate}
  \acro{SES}{Simple Exponential Smoothing}
  \acro{S-GW}{Serving Gateway}
  \acro{SGD}{\LU{S}{s}tochastic \LU{G}{g}radient \LU{D}{d}escent}  
  \acro{SINR}{signal-to-interference-plus-noise ratio}
  \acro{SI}{self-interference}
  \acro{SIP}{Session Initiation Protocol}
  \acro{SISO}{\LU{S}{s}ingle \LU{I}{i}nput \LU{S}{s}ingle \LU{O}{o}utput}
  \acro{SIMO}{Single Input Multiple Output}
  \acro{SIR}{Signal to Interference Ratio}
  \acro{SLNR}{Signal-to-Leakage-plus-Noise Ratio}
  \acro{SMA}{Simple Moving Average}
  \acro{SNR}{\LU{S}{s}ignal-to-\LU{N}{n}oise \LU{R}{r}atio}
  \acro{SORA}{Satisfaction Oriented Resource Allocation}
  \acro{SORA-NRT}{Satisfaction-Oriented Resource Allocation for Non-Real Time Services}
  \acro{SORA-RT}{Satisfaction-Oriented Resource Allocation for Real Time Services}
  \acro{SPF}{Single-Carrier Proportional Fair}
  \acro{SRA}{Sequential Removal Algorithm}
  \acro{SRS}{sounding reference signal}
  \acro{SU-MIMO}{Single-User Multiple Input Multiple Output}
  \acro{SU}{Single-User}
  \acro{SVD}{Singular Value Decomposition}
  \acro{SVM}{\LU{S}{s}upport \LU{V}{v}ector \LU{M}{m}achine}
  \acro{TCP}{Transmission Control Protocol}
  \acro{TDD}{time division duplex}
  \acro{TDMA}{\LU{T}{t}ime \LU{D}{d}ivision \LU{M}{m}ultiple \LU{A}{a}ccess}
  \acro{TNFD}{three node full duplex}
  \acro{TETRA}{Terrestrial Trunked Radio}
  \acro{TP}{Transmit Power}
  \acro{TPC}{Transmit Power Control}
  \acro{TTI}{transmission time interval}
  \acro{TTR}{Time-To-Rendezvous}
  \acro{TSM}{Throughput-based Satisfaction Maximization}
  \acro{TU}{Typical Urban}
  \acro{UE}{\LU{U}{u}ser \LU{E}{e}quipment}
  \acro{UEPS}{Urgency and Efficiency-based Packet Scheduling}
  \acro{UL}{uplink}
  \acro{UMTS}{Universal Mobile Telecommunications System}
  \acro{URI}{Uniform Resource Identifier}
  \acro{URM}{Unconstrained Rate Maximization}
  \acro{VR}{Virtual Resource}
  \acro{VoIP}{Voice over IP}
  \acro{W-MAC}{\LU{W}{w}ireless \LU{M}{m}ultiple \LU{A}{a}ccess \LU{C}{c}hannel}
  \acro{WAN}{Wireless Access Network}
  \acro{WCDMA}{Wideband Code Division Multiple Access}
  \acro{WF}{Water-filling}
  \acro{WiMAX}{Worldwide Interoperability for Microwave Access}
  \acro{WINNER}{Wireless World Initiative New Radio}
  \acro{WLAN}{Wireless Local Area Network}
  \acro{WMMSE}{weighted minimum mean square error}
  \acro{WMPF}{Weighted Multicarrier Proportional Fair}
  \acro{WPF}{Weighted Proportional Fair}
  \acro{WSN}{Wireless Sensor Network}
  \acro{WWW}{World Wide Web}
  \acro{XIXO}{(Single or Multiple) Input (Single or Multiple) Output}
  \acro{ZF}{Zero-Forcing}
  \acro{ZMCSCG}{Zero Mean Circularly Symmetric Complex Gaussian}
\end{acronym}

\title{Majority Vote Compressed Sensing}

\author{Henrik~Hellström\IEEEauthorrefmark{2}\IEEEauthorrefmark{1},~\IEEEmembership{Member,~IEEE,}
Jiwon~Jeong\IEEEauthorrefmark{1},~\IEEEmembership{Graduate Student Member,~IEEE,}
Ayfer~Özgür\IEEEauthorrefmark{1},~\IEEEmembership{Member,~IEEE,}
Viktoria~Fodor\IEEEauthorrefmark{2},~\IEEEmembership{Member,~IEEE,}
and~Carlo~Fischione\IEEEauthorrefmark{2},~\IEEEmembership{Fellow,~IEEE}\\
\thanks{We acknowledge the financial support of the Ericsson Research Foundation, the Swedish Foundation for Strategic Research: SAICOM, and KTH Digital Futures.}
\IEEEauthorblockA{\IEEEauthorrefmark{1}Electrical Engineering Department, Stanford University, California, USA\\
Emails: hhells@stanford.edu, jeongjw@stanford.edu, aozgur@stanford.edu\\}
\IEEEauthorblockA{\IEEEauthorrefmark{2}School of Electrical Engineering and Computer Science, KTH Royal Institute of Technology, Stockholm, Sweden\\
Emails: hhells@kth.se, vjfodor@kth.se, carlofi@kth.se}
}

\markboth{IEEE Transactions on Wireless Communications, submitted for review}%
{Henrik}


\maketitle

\begin{abstract}
We consider the problem of non-coherent over-the-air computation (AirComp), where $n$ devices carry high-dimensional data vectors $\vtX_i\in\mathbb{R}^d$ of sparsity $\lVert\vtX_i\rVert_0\leq k$ whose sum has to be computed at a receiver. Previous results on non-coherent AirComp require more than $d$ channel uses to compute functions of $\vtX_i$, where the extra redundancy is used to combat non-coherent signal aggregation. However, if the data vectors are sparse, sparsity can be exploited to offer significantly cheaper communication. In this paper, we propose to use random transforms to transmit lower-dimensional projections $\vtS_i\in\mathbb{R}^T$ of the data vectors. These projected vectors are communicated to the receiver using a majority vote (MV)-AirComp scheme, which estimates the bit-vector corresponding to the signs of the aggregated projections, i.e., $\vtY = \text{sign}(\sum_i\vtS_i)$. By leveraging 1-bit compressed sensing (1bCS) at the receiver, the real-valued and high-dimensional aggregate $\sum_i\vtX_i$ can be recovered from $\vtY$. We prove analytically that the proposed MVCS scheme estimates the aggregated data vector $\sum_i \vtX_i$ with $\ell_2$-norm error $\epsilon$ in $T=\mathcal{O}(kn\log(d)/\epsilon^2)$ channel uses. Moreover, we specify algorithms that leverage MVCS for histogram estimation and distributed machine learning. Finally, we provide numerical evaluations that reveal the advantage of MVCS compared to the state-of-the-art.
\end{abstract}


\section{Introduction}
\Ac{AirComp} is a wireless communication method that leverages signal superposition to communicate mathematical functions. Compared to current communication systems, \ac{AirComp} is poised to reduce latency or spectrum use by a factor proportional to the number of simultaneously transmitting devices. Therefore, \ac{AirComp} has the potential to offer cheap and efficient collection of distributed data in dense wireless networks. Wireless data collection has numerous current and upcoming applications, such as federated analytics, sensor networks, and distributed control systems \cite{csahin2023survey}. This is also increasingly relevant for distributed machine learning, as the growing size of \ac{ML} models necessitates training across many servers \cite{zhang2015stac}. However, while \ac{AirComp} is a promising technology for these applications, it is also associated with significant technical challenges that, so far, make it difficult to implement it in practice.

One of the main challenges of \ac{AirComp} is the lack of reliable error correction. Consider that $n$ devices each carry a message $\vtX_i \in \mathbb{R}^d,\, \forall i\in\{1,...,n\}$ and that we are interested in recreating some function $f(\vtX_1,...,\vtX_n)$ of these messages at an \ac{AP}. With traditional orthogonal communication, these $n$ messages can be communicated almost perfectly to the \ac{AP} by leveraging forward error correcting codes and retransmissions triggered by cyclic redundancy check. The desired function $f(\vtX_1,...,\vtX_n)$ can then be computed at the \ac{AP} with very high levels of reliability. However, in \ac{AirComp}, instead of transmitting these messages separately, all $n$ devices transmit over the same radio resources, and the superimposed signal is used to estimate $f(\vtX_1,...,\vtX_n)$. Since the \ac{AP} never sees individual signals from the devices, error correction and detection are challenging \cite{you2023broadband}. In the current state-of-the-art, an \ac{AirComp} receiver does not know if the function is computed correctly. Instead, such systems generally provide estimates $\hat{f}$ of the desired function, with some residual function estimation error $\epsilon=\lVert\hat{f}(\vtX_1,...,\vtX_n)-f(\vtX_1,...,\vtX_n)\rVert_2$.

For a system designer considering \ac{AirComp}, the estimation error $\epsilon$ can be viewed as a price to pay for the promised latency reduction and spectral efficiency gain. Ideally, this should not be a binary choice, but \ac{AirComp} should be tunable to trade off communication efficiency and $\epsilon$, similar to the role of the code rate in digital communications.

To achieve such a tradeoff several authors have considered the use of \ac{CS} \cite{amiri2020machine, fan20221, edin2024over}. The basic idea of compressed sensing is to project the high-dimensional messages $\vtX_i\in\mathbb{R}^d$ to lower-dimensional vectors $\vtS_i\in\mathbb{R}^T$ by applying a linear transform $\vtS_i \triangleq \mathbf{M}\vtX_i$ at the user devices before transmission. Given that the messages $\vtX_i$ satisfy sparsity constraints $\Vert \vtX_i \Vert_0 \leq k, \forall i \in \{1, \dots, n\}$ and the matrix $\mathbf{M}$ fulfills the restricted isometry property \cite{candes2008introduction}, the messages $\vtX_i$ can be exactly recreated from $\vtS_i$ at the receiver. In the context of \ac{AirComp}, the receiver will not have access to the individual vectors $\vtS_i$ but rather a noisy estimate of $f(\vtS_1,...,\vtS_n)$. Still, the desired function $f(\vtX_1, ..., \vtX_n)$ can be recreated from $\hat{f}(\vtS_1,...,\vtS_n)$, with an estimation error that depends on the dimension $T$ \cite{edin2024over}. As such, the parameter $T$ can be tuned to provide the desired tradeoff.


While \ac{CS} has gotten prior attention in the context of \ac{AirComp}, the state-of-the-art papers make significant simplifications in the channel model, either by limiting the distortions to \ac{AWGN} \cite{amiri2020machine} or by assuming perfect channel inversion at the transmitter-side \cite{amiri2020federated, fan20221, edin2024over}. Both of these simplifications ignore the need for transmitter side phase correction, which is known to be challenging in
practice, since it requires synchronization of the devices' oscillators \cite{goldenbaum2013robust, csahin2023distributed, you2023broadband}. In this paper, we bridge this gap by developing \ac{CS}-based \ac{AirComp} that does not require phase correction at the transmitters. Our main idea is to combine the recently proposed \ac{MV} \ac{AirComp}~\cite{csahin2023distributed} scheme with \ac{1bCS}~\cite{boufounos2008}.

In \ac{MV}-\ac{AirComp}, the $n$ vectors $\vtX_i\in\mathbb{R}^d$ are communicated in $2d$ orthogonal radio resources, with the goal to compute the majority vote, defined as $\vtF = \text{sign}(\sum_i \vtX_i)$. For each element $x_i[t]$ (where $t\in\{1,...,d\}$) of the vector, the devices decide to communicate in one of two orthogonal resources, depending on the sign of $x_i[t]$. By comparing the received power over the two resources, the \ac{AP} can determine the majority vote, effectively collecting one bit of information from $n$ devices. This way, by leveraging $2d$ orthogonal radio resources, all $n$ devices can communicate $d$ bits of information, leading to a spectral efficiency of $n/2$ bits per channel use\footnote{The outcome of the majority vote only consists of $d$ bits of information, but $nd$ bits are used to compute it.}. However, noise in the wireless channel can result in bit errors with \ac{MV}-\ac{AirComp}, which in turn can lead to large estimation error for $\vtF$. To combat these errors, we employ \ac{1bCS}. \ac{1bCS} is a dimensionality reduction method that recreates a real-valued vector $\vtX_*\in\mathbb{R}^d$, using a compressed bit-vector $\vtB\in\{0,1\}^T$, where $T<d$ \cite{boufounos2008}. It allows us to reduce the number of channel uses needed significantly below $d$ in the regime where $n \ll d$ and $k \ll d$. This is the regime of interest for many federated learning and analytics applications. As we discuss later in Section~\ref{sec:background}, \ac{1bCS} is robust to errors under basic assumptions. The combination of these ideas allows the proposed \ac{MVCS} scheme to efficiently and reliably compute real-valued functions without posing any demands on phase synchronization.

The contributions of our paper can be summarized as follows:
\begin{itemize}
    \item We propose, as far as we are aware, the first non-coherent \ac{AirComp} scheme with compressed sensing.
    \item We prove that \ac{MVCS} can achieve any desired estimation error $\epsilon$ in $T=\mathcal{O}\left(kn\log(d)/\epsilon^2\right)$ channel uses, where $d$ is the dimension of the desired function $\sum_i\vtX_i\in\mathbb{R}^d$ and $\lVert\sum_i\vtX_i\rVert_0\leq kn$. 
    \item We demonstrate that \ac{MVCS} can be leveraged for private histogram estimation, without having access to individual user item. 
    \item We provide numerical results for training deep neural networks on the MNIST problem over wireless channels that incorporate path loss, fading, and \ac{AWGN}. The numerical results suggest that \ac{MVCS} can outperform state-of-the-art schemes in distributed machine learning. 
\end{itemize}
\section{Background}
\label{sec:background}
\subsection{1-bit Compressed Sensing}\label{sec:1bCS}
In the original \ac{1bCS} framework, the goal is to estimate a high-dimensional sparse target vector $\vtX_*\in\mathbb{R}^d$ using lower-dimensional and quantized 1-bit measurements of $\vtX_*$. In particular, the measurement process is modeled as a matrix-vector multiplication $\vtS = \mathbf{M}\vtX_*$, where $\mathbf{M}\in\mathbb{R}^{T\times d}$ is called the measurement matrix\footnote{In compressed sensing (not the 1-bit version), exact target reconstruction is possible given that the measurement matrix satisfies the restricted isometry property, which holds for many random matrix ensembles \cite{candes2006stable}. However, \ac{1bCS} is more restrictive, since certain discrete distributions, such as Bernoulli matrices, will yield indistinguishable measurements $\vtS=\vtS'$ for distinct pairs of target vectors $\vtX\neq\vtX'$. Even so, there exists some random matrix constructions, such as the gaussian matrix, that yield exact reconstruction in \ac{1bCS} with high probability \cite{plan2012robust}.}. A common choice is to let the elements of $\mathbf{M}$ be i.i.d. standard Gaussian \acp{RV}. After compression, the measurement $\vtS$ is quantized down to a single bit of information per element $t\in\{1, \hdots, T\}$ as
\begin{equation}\label{eq:quantization}
    b[t] = \sign\left(s[t]\right) \triangleq \begin{cases}
        +1,   & \text{if } s[t] \geq 0 \\
        -1,  & \text{if } s[t] < 0,
    \end{cases}
\end{equation}
where the choice for the tie-breaker of $s[t] = 0$ is arbitrary. Such a measure yields two types of compression. Firstly, a reduction in dimension since $\mathbf{M}$ can be constructed such that $T\ll d$ and secondly a quantization gain down to 1 bit per element. However, note that the quantized measure $\vtB$ does not preserve any information about the magnitude of $\vtS$. As such, it is not possible to recreate the norm of $\vtX_*$ from $\vtB$ but only the relative magnitudes of its elements. Therefore, it is common practice to assume that $\lVert\vtX_*\rVert_2=1$ and force the estimate $\hat{\vtX}\in\mathbb{R}^d$ to be a unit $\ell_2$-norm vector as well. Given this setup, the goal of \ac{1bCS} is to develop an algorithm $\mathcal{A}\left(\vtB; \mathbf{M}\right) \rightarrow \hat{\vtX}$ that can guarantee
\begin{equation}\label{eq:reconstruction_guarantee}
    \lVert\vtX_*-\hat{\vtX}\rVert_2 \leq \epsilon
\end{equation}
with some probability $p$. There are many published algorithms of this sort \cite{boufounos2008, boufounos2009greedy, plan2013one, plan2012robust, zhang2014efficient, xiao20191, xiao2019one, genzel2020robust, matsumoto2024robust}, with variants in the system model, reconstruction performance, and/or computational complexity. The extension of the system model to noisy measurements is of particular importance to our work \cite{plan2012robust, zhang2014efficient, xiao20191, matsumoto2024robust}. A simple way to model noisy measurements is to assign a random probability that each bit is flipped in \eqref{eq:quantization} but there is also a more general model introduced in \cite{plan2012robust}. In particular, they consider that the quantized bits are drawn independently at random and that their expected value is given by
\begin{equation}\label{eq:b_noise}
    \theta\left( \langle \vtM_t,\vtX_* \rangle \right)\triangleq\mathbb{E}\left[b[t]|\vtM_t\right],
\end{equation}
where $\vtM_t\in\mathbb{R}^{1\times d}$ is row $t$ of the measurement matrix, $\langle.,.\rangle$ is the inner product, and $\theta$ is some function. A key result in \ac{1bCS} is that this function $\theta$ does not need to be known to estimate $\vtX_*$ \cite{plan2012robust}. Given that $\mathbf{M}$ is a random Gaussian matrix, the only required assumptions on $\theta$ are that $-1\leq\theta(\cdot)\leq1$ and
\begin{equation}\label{eq:lambda_def}
    \mathbb{E}\left[\theta(g)g\right] \triangleq \lambda_g > 0,
\end{equation}
where $g$ is a standard Gaussian variable. This assumption ensures that the 1-bit measurements $\vtB$ are positively correlated with the measurements $\vtS$. Generally, the error $\epsilon$ in \eqref{eq:reconstruction_guarantee} is inversely proportional to $\lambda_g$, indicating that a higher correlation yields better reconstruction results.

Note that \eqref{eq:lambda_def} comes from the compressed sensing literature, and not from the \ac{AirComp} literature. In such setups, there is only one original vector $\vtX_*$. Therefore, \cite{plan2012robust} can guarantee that $\langle \vtM_t,\vtX_* \rangle$ is standard Gaussian by normalizing $\vtX_*$ to have unit norm. In the \ac{AirComp} case, we are interested in recreating $\vtF\coloneq\sum_i\vtX_i$. We can select the norm of each $\vtX_i$ by performing pre-processing at the sensor devices, but $\lVert\sum_i\vtF\rVert$ cannot be controlled. Therefore, in this paper, we define $\lambda$ as follows.
\begin{equation} \label{eq:lambda_def_2}
    \mathbb{E}\left[\theta(\langle \vtM_t,\vtF \rangle) \langle \vtM_t,\vtF \rangle \right] \triangleq \lambda.
\end{equation}
The \ac{1bCS} method still works under the same requirement on $\theta$; $\lambda$ must be greater than 0.

\subsection{Passive Algorithm}
Most \ac{1bCS} algorithms are iterative algorithms based on linear programming \cite{plan2013one}, convex programming \cite{plan2012robust}, first-order algorithms \cite{matsumoto2024robust}, matching pursuit \cite{boufounos2009greedy}, or variants of gradient descent \cite{boufounos2008, xiao20191, xiao2019one}. However, in \cite{zhang2014efficient}, the authors proposed a passive algorithm that estimates $\vtX_*$ in one shot with a closed-form solution. In addition to being computationally efficient, the passive algorithm has the advantage that it does not impose an exact sparsity $\lVert\hat{\vtX}\rVert_0$. This is important in our application because the receiver only knows an upper bound on the sparsity of the target, while the exact value of $\lVert\vtF\rVert_0$ is unknown. For a complete description of the passive algorithm, we refer to \cite{zhang2014efficient}, but we also give a brief exposition here for completeness and to unify notation.

The passive algorithm leverages the soft-thresholding operator, defined as follows.
\begin{definition}[Soft-Thresholding Operator]\label{def:soft_thresh}
The soft-thresholding operator with argument $\alpha\in\mathbb{R}$ and threshold $\gamma\in\mathbb{R}^+$ is defined as
   \begin{equation}\label{eq:soft_thresh}
    P_{\gamma}(\alpha) \triangleq \begin{cases}
        0, & \text{if } |\alpha| \leq \gamma \\
        \sign\left(\alpha\right)\left(|\alpha| - \gamma\right), & \text{otherwise}
    \end{cases},
\end{equation}
where $\sign(\cdot)$ was defined in \eqref{eq:quantization}. We extend the operator to vector-valued arguments $\vtAlpha\in\mathbb{R}^d$ as
\begin{equation}
    P_{\gamma}(\vtAlpha) \triangleq \left[P_{\gamma}(\alpha[1]),\, P_{\gamma}(\alpha[2]),\, \dots,\, P_{\gamma}(\alpha[d])\right]^H,
\end{equation}
where $(\cdot)^H$ is the transpose, i.e., $P_{\gamma}(\vtAlpha)$ is a column vector.
\end{definition}
Given this definition, the closed-form estimate is computed as
\begin{equation}\label{eq:x_hat}
    \hat{\vtX} = \begin{cases}
        \vtZero, & \text{if } \lVert \frac{1}{T}\mathbf{M}^H\vtB \rVert_{\infty} \leq \gamma, \\
        \frac{1}{\lVert P_{\gamma}\left(\frac{1}{T}\mathbf{M}^H\vtB\right)\rVert_2}P_{\gamma}\left(\frac{1}{T}\mathbf{M}^H\vtB\right) & \text{otherwise},
    \end{cases}
\end{equation}
where $\mathbf{M}$ is the measurement matrix.
With this algorithm, the probability of estimating $\vtX$ within a given error margin $\epsilon$ increases exponentially in $T$, as specified in the following theorem.
\begin{theorem}\label{th:zhang}
    Assume
    \begin{equation}\label{eq:gamma}
        \gamma = 2c\sqrt{\frac{\delta + \log(d)}{T}}
    \end{equation}
    for some constant $c$. Consider $\delta>0$ to be a hyperparameter that can be chosen freely. Additionaly, consider that the noise process defined in \eqref{eq:b_noise} satisfies $\lambda_g>0$. Let $\vtX_*$ be sparse and $\hat{\vtX}$ be selected according to \eqref{eq:x_hat}. Then, with a probability of at least $1-e^{1-\delta}$, we have
    \begin{equation}\label{eq:error}
        \lVert\hat{\vtX}-\vtX_*\rVert_2 \leq \epsilon \triangleq \frac{3\gamma}{\lambda_g}\sqrt{\lVert\vtX_*\rVert_0}.
    \end{equation}
\end{theorem}
\begin{proof}
    See \cite{zhang2014efficient}.
\end{proof}
In other words, in order to guarantee that our estimate $\hat{\vtX}$ is within an $\ell_2$-norm ball of radius $\epsilon$ to the target vector, we need the measurement's dimension $T$ to be in the order of
\begin{equation}\label{eq:communication_cost}
    T = O\left( \frac{\lVert\vtX_*\rVert_0}{(\lambda_g\epsilon)^2}\left(\log(d) + \delta\right) \right).
\end{equation}

\begin{remark}
    Recently, \cite{matsumoto2024robust} proposed a new \ac{1bCS} algorithm that provides a similar guarantee for $T=O\left( \frac{\lVert\vtX_*\rVert_0}{\epsilon}\log(d) \right)$, a factor $\epsilon$ improvement compared to \cite{zhang2014efficient}'s algorithm. However, \cite{matsumoto2024robust} has two downsides compared to \cite{zhang2014efficient}. Firstly, the computational cost is significantly higher, since many iterations can be required to converge. Secondly, in every iteration, the algorithm forces the vector to have an exact sparsity, which might be unknown at the receiver. If the computational capability at the receiver is sufficient and $\lVert\vtF\rVert_0$ is known, \cite{matsumoto2024robust}'s algorithm can be used in place of the passive algorithm without major changes to \ac{MVCS}.
\end{remark}

\subsection{Over-the-Air Computation}
The core idea of \ac{AirComp} is to leverage the electromagnetic superposition of simultaneously transmitted signals to compute a function \cite{goldenbaum2013robust}. This superposition is generally modeled using the \ac{W-MAC} \cite{goldenbaum2013robust}, which is defined as follows.

\begin{definition}[W-MAC]\label{def:W-MAC}
    For any channel use $t\in\mathbb{Z}^+$, the \ac{W-MAC} is a map from $\mathbb{C}^n$ to $\mathbb{C}$ defined to be
    \begin{equation}\label{eq:wmac}
        \left(a_1[t], \hdots, a_n[t]\right) \rightarrow \sum_{i=1}^nl_ih_i[t]a_i[t]+z[t] \triangleq y[t],
    \end{equation}
    where $l_i\in\mathbb{R}$ denotes large-scale fading coefficients, $h_i[t]\in\mathbb{C}$ denotes small-scale fading coefficients, $a_i[t]\in\mathbb{C}$ denotes the transmitted I/Q symbols, and $z[t]\in\mathbb{C}$ denotes \ac{AWGN}. When applicable, $t$ is omitted for brevity.
\end{definition}

From Definition \ref{def:W-MAC} it is clear that the \ac{W-MAC} naturally computes a sum of the transmitted messages. However, due to the phase rotations $\angle h_i[t]$, the superposition can be destructive. Moreover, there is an issue of heterogeneous signal attenuation, which adds real-valued weights $l_i|h_i[t]|$ to the message amplitudes. As wireless researchers, we are used to dealing with these types of distortions, but compared to orthogonal communication schemes, it is significantly more challenging to compensate for the attenuation and phase rotation in \ac{AirComp}. In particular:
\begin{enumerate}
    \item \textbf{Channel Estimation}: It is not straightforward to attach a preamble to an \ac{AirComp} message for channel estimation. Consider that a known preamble is attached to the beginning of each message $a_i$. When the receiver samples $y[t]$ from \eqref{eq:wmac}, the non-coherent addition causes the received pilot to differ significantly from the transmitted ones, making the channel estimation problem challenging. Instead, orthogonal resources could periodically be allocated to estimate channels, but then the communication efficiency benefits of \ac{AirComp} are substantially reduced \cite{abari2016over}.
    \item \textbf{Phase and Frequency tracking:} To avoid phase rotations of the received symbols, it is standard practice to compensate for frequency and phase offsets by aligning the receiver's local oscillator with the incoming signal. In schemes such as IEEE 802.11a, this is achieved by sending dedicated pilot tones that act as a reference for a phase-locked loop at the receiver \cite[Section 27.3.11.10]{ieee80211ax}. However, in \ac{AirComp} the signals are superimposed, so the receiver's local oscillator cannot compensate for the non-coherent superposition of the $n$ waves at the passband. Instead, coherent superposition requires that the \acp{LO} of all $n$ transmitters are synchronized. To put this into perspective, if the carrier waves are in the GHz range, the $n$ transmitters need to be synchronized well within the carrier period $T_c = 1/f_c \leq 1\text{ns}$. The 5G standard, at best, guarantees $260\text{ns}$ time synchronization~\cite{5gspecification}.
\end{enumerate}
For these reasons, there is a significant number of articles on non-coherent \ac{AirComp} schemes \cite{goldenbaum2013robust, csahin2023over, mohammadi2019collaborative, csahin2023distributed, sery2020analog}. Some schemes consider that the magnitude of the channel information is available, i.e., the devices can correct for $l_i|h_i[t]|$ but are unable to synchronize phase \cite{goldenbaum2013robust, csahin2023over, mohammadi2019collaborative}. Others operate with even more stringent assumptions, where neither magnitude nor phase can be corrected at the devices \cite{csahin2023distributed, sery2020analog}. The system proposed in this paper assumes no device-side knowledge of $|h_i[t]|$ nor ability to synchronize phase. However, we assume that large-scale fading $l_i$ is known by the network, to facilitate power control.


\section{Majority Vote 1-bit Compressed Sensing}\label{sec:proposal}
In this section, we provide an application-agnostic description of \ac{MVCS}. Note that the \ac{MV}-\ac{AirComp} scheme is not novel \cite{csahin2023distributed}. Our contribution is the combination of \ac{MV}-\ac{AirComp} with \ac{1bCS} and the associated analysis to show \ac{AirComp} can reach arbitrarily small estimation errors under practical network assumptions.

\subsection{System Model and Assumptions}\label{sec:system_model}
Consider a wireless network with $n$ user devices, each carrying a message $\vtX_i\in\mathbb{R}^d$. Further, consider that the messages are sparse $\lVert\vtX_i\rVert_0\leq k$ and of unit $\ell_2$-norm $\lVert\vtX_i\rVert_2=1$\footnote{Both these assumptions can be made true by applying top-$k$ sparsification and normalizing before transmission.}. The goal of the network is to recover the aggregate $\vtF \triangleq \sum_i\vtX_i$ at an \ac{AP}, by communicating over the \ac{W-MAC} channel model, as defined in Definition \ref{def:W-MAC}. To enable \ac{1bCS}, we consider that every device (including the \ac{AP}) share a common matrix $\mathbf{M}\in\mathbb{R}^{T\times d}$. The matrix is generated by independent sampling from a standard Gaussian distribution. Once sampled, $\mathbf{M}$ remains static. This matrix can be large but, in practice, it can be generated on all devices independently as long as they agree on a random seed.

The devices communicate I/Q symbols $\vtA_i\in\mathbb{C}^{2T}$ over the \ac{W-MAC}, where power control is performed to satisfy the average power constraint
\begin{equation}\label{eq:power_constraint}
        \mathbb{E}_{a_i[t]}\left[|a_i[t]|^2 \vert x_i[t]\right] \leq \overline{P},\,\forall\,i
\end{equation}
where $\overline{P}$ is common to all devices. Throughout the paper, all expectations on the transmitted symbols $a_i[t]$ are conditioned on $\vtX_i$, since we do not wish to impose any known distribution on the transmitted messages. We may omit the explicit expectation for brevity.

We model the large-scale fading as
\begin{equation}\label{eq:pathloss}
    l_i=r_i^{-\frac{\beta}{2}},
\end{equation}
where $r_i$ is the distance from the \ac{AP} to \ac{UE} $i$ and $\beta$ is the path loss exponent. We assume that the large-scale fading remains constant for the duration of at least $2T$ transmissions and that all $l_i$ are known by the \ac{AP}. For the fast-fading channel coefficients, we consider Rayleigh fading, i.e., $h_i[t]\sim\mathcal{C}\mathcal{N}\left(0,\sigma_h^2\right)$ for all $t$. Further, we assume independence between devices $\mathbb{E}[h_i[t]h_j[t]]=0\,\forall i\neq j$ and channel uses $\mathbb{E}[h_i[t]h_i[s]]=0\,\forall t\neq s$, and we consider that there is no \ac{CSI} available.

\subsection{Transmitter Baseband}\label{sec:transmitter}
Before transmission begins, all devices generate a measure $\vtS_i\in\mathbb{R}^T$ as
\begin{align}
    \vtS_i = \mathbf{M}\vtX_i.
\end{align}
To transmit these measures, the \ac{AP} allocates $2T$ orthogonal radio resources that are shared by all $n$ devices, which yields two channel uses for each element of $\vtS_i$. These two channel uses are leveraged to communicate negative and positive measures, respectively. Device $i\in\{1, \hdots, n\}$ will only activate one of the two resources for each index $t\in\{1, \hdots, T\}$, depending on the sign of $s_i[t]$. This way, positive and negative values are encoded in the power of their respective radio resources, rather than in the phase of the IQ symbols \cite{csahin2023distributed}, which is usually how it is implemented in coherent over-the-air computation. We define that channel uses $\{1,\hdots, T\}$ and $\{T+1,\hdots,2T\}$ are used to communicate negative and positive values, respectively. For index $t$, device $i$ will transmit
\begin{align}\label{eq:transmit_IQ}
    a_i[t] &= \sqrt{p_i|s_i[t]|}u\left(-s_i[t]\right), \\
    a_i[t+T] &= \sqrt{p_i|s_i[t]|}u\left(s_i[t]\right), \label{eq:transmit_IQ_2}
\end{align}
where $u(\cdot)$ is the Heaviside step function and $p_i\in\mathbb{R}$ is a power control factor. The power control factor is used to align the received power of the devices and to satisfy the power constraint \eqref{eq:power_constraint}. The selection of $p_i$ is outlined in the following proposition. 
\begin{proposition}\label{prop:power_control}
    Let the measurement matrix $\mathbf{M}$ be a standard Gaussian matrix, let $l_{\min} \triangleq \min_i l_i$, and let $\vtX_i$ be a $k$-sparse vector with unit $\ell_2$-norm. Then, the power control factor
    \begin{equation}\label{eq:tx_power}
        p_i = \sqrt{2\pi}\frac{\overline{P}l_{\min}^2}{l_i^2}
    \end{equation}
    satisfies the average power constraint
    \begin{equation}
        \mathbb{E}_{a_i}\left[|a_i[t]|^2\right] \leq \overline{P},\,\forall\,i,
    \end{equation}
    and aligns the received powers as
    \begin{equation}\label{eq:rx_align}
    \begin{split}
        &\mathbb{E}\left[|l_ih_i[t]a_i[t]|^2\right] =\\
        &\mathbb{E}_{h_i[t],s_i[t]}\left[|l_ih_i[t]\sqrt{p_is_i[t]}u(s_i[t])|^2\right] = \overline{P}\sigma_h^2l_{\min}^2\,\forall\,i.
    \end{split}
    \end{equation}
\end{proposition}
\begin{proof}
    The average transmit power is given by
    \begin{equation}\label{eq:ex_power}
        \mathbb{E}_{a_i}\left[|a_i[t]|^2\right] = p_i\mathbb{E}_{s_i[t]}\left[|s_i[t]|u(s_i[t])\right] \underset{(a)}{=} \frac{p_i}{2}\mathbb{E}_{s_i}\left[|s_i[t]|\right],
    \end{equation}
    where $(a)$ holds since $\mathbb{E}_{s_i}[u(s_i[t])]=1/2$ ($s_i[t]$ is symmetric around zero) and $|s_i[t]|$ is uncorrelated to $u(s_i[t])$. Without loss of generality, consider that the $k$ non-zero elements of $\vtX_i$ are the first $k$ elements $x_i[j]$ where $j\in\{1,\hdots,k\}$. The measure $\vtS_i$ can then be expressed as
    \begin{equation}
        \vtS_i = \begin{pmatrix}
            \sum_{j=1}^km_1[j]x_i[j] \\
            \sum_{j=1}^km_2[j]x_i[j] \\
            \vdots \\
            \sum_{j=1}^km_T[j]x_i[j] \\
        \end{pmatrix},
    \end{equation}
    where $m_t[j]$ is element $j$ of row $t$ of $\mathbf{M}$. Since the elements of $\mathbf{M}$ are i.i.d. standard Gaussian \acp{RV}, we have
    \begin{equation}\label{eq:s_ex}
    \begin{split}
        \mathbb{E}\left[|s_i[t]|\right] &=  \mathbb{E}\left[\left|\sum_{j=1}^km_t[j]x_i[j]\right|\right] = \mathbb{E}\left[\left|G\right|\right] \\
        &\underset{(a)}{=} \frac{\sqrt{2}\lVert\vtX_i\rVert_2}{\sqrt{\pi}} \underset{(b)}{=} \sqrt{\frac{2}{\pi}},
    \end{split}
    \end{equation}
    where $G$ is a zero-mean Gaussian \ac{RV} with variance $\sigma_g^2=\lVert\vtX_i\rVert_2^2$. Equality $(a)$ holds because $|G|$ follows the half-normal distribution and equality $(b)$ holds since we assume that all messages $\vtX_i$ have unit $\ell_2$-norms. Combining \eqref{eq:s_ex} with \eqref{eq:ex_power} and \eqref{eq:tx_power} yields
    \begin{equation}\label{eq:ex_w}
        \mathbb{E}_{a_i}\left[|a_i[t]|^2\right] = \frac{\overline{P}l_{\min}^2}{l_i^2} \underset{(a)}{\leq} \overline{P}.
    \end{equation}
    Note that power constraint $(a)$ in \eqref{eq:ex_w} is only active for the device furthest from the \ac{AP} and that all other devices will on average communicate with less power to achieve alignment. The equality $\mathbb{E}\left[|l_ih_i[t]a_i[t]|^2\right]=\overline{P}\sigma_h^2l_{\min}^2$ follows by combining \eqref{eq:ex_w} with the left-hand side of \eqref{eq:rx_align}.
\end{proof}

\subsection{Receiver Baseband}\label{sec:receiver}
After the user devices transmit $\vtA_i$ by accessing the \ac{W-MAC} $2T$ times, the receiver samples the symbols $\vtY\in\mathbb{C}^{2T}$, where symbol $t\in\{1,\hdots,T\}$ and $t+T$ are
\begin{equation}\label{eq:rx_samples}
\begin{split}
    y[t] &= E\sum_{i=1}^nh_i[t]\sqrt{|s_i[t]|}u\left(-s_i[t]\right) + z[t] \text{ and }\\
    y[t+T] &= E\sum_{i=1}^nh_i[t+T]\sqrt{|s_i[t]|}u\left(s_i[t]\right) + z[t+T],
\end{split}
\end{equation}
where $E\triangleq(2\pi)^{1/4}\sqrt{\overline{P}}l_{\min}$, which is the result of selecting $p_i$ according to Proposition \ref{prop:power_control}. The values $y[t]$ and $y[t+T]$ are interpreted as votes (weighted by $s[i]$) for whether the aggregate measure $\sum_{i}s_i[t]$ is positive or negative, similar to the \ac{MV} \ac{AirComp} scheme proposed in \cite{csahin2023distributed}. The receiver resolves the vote by a simple comparison of the energies in the two radio resources as
\begin{equation}\label{eq:rx_quantize}
    b[t] = \sign\left(|y[t+T]|^2-|y[t]|^2\right),
\end{equation}
where $\sign(\cdot)$ was defined in \eqref{eq:quantization}. This demodulated vector $\vtB\in\{-1,1\}^T$ follows the noise model from \eqref{eq:b_noise} and is therefore compatible with the \ac{1bCS} framework. We prove this statement in the following theorem.
\begin{theorem}\label{th:mvcs}
    Consider that the user devices transmit $\vtA_i\in\mathbb{C}^{2T}$ based on \eqref{eq:transmit_IQ} and \eqref{eq:transmit_IQ_2} with the power control scheme from Proposition \ref{prop:power_control}. Further consider that the receiver samples $\vtY\in\mathbb{C}^{2T}$ via the \ac{W-MAC} defined in Definition \ref{def:W-MAC} and then forms the received bit-vector $\vtB\in\{-1,1\}^T$ according to \eqref{eq:rx_quantize}. Let
    \begin{equation}\label{eq:th2_theta}
        \theta\left( \langle \vtM_t,\vtF \rangle \right) \triangleq \mathbb{E}\left[b[t]|\vtM_t\right],
    \end{equation}
    where $\vtF \triangleq \sum_i\vtX_i$. Then, we have that 
    \begin{equation}
        \theta(\langle\vtM_t,\vtF\rangle) = \frac{E^2\sigma_h^2\langle\vtM_t,\vtF\rangle}{E^2\sigma_h^2\sum_{i=1}^n\left|\langle\vtM_t,\vtX_i\rangle\right| + 2\sigma_z^2}, \,
    \end{equation}
    and $\theta(\cdot)$ satisfies the positive correlation property
    \begin{equation}\label{eq:positive_correlation}
        \mathbb{E}\left[\theta(\langle \vtM_t,\vtF \rangle) \langle \vtM_t,\vtF \rangle \right] > 0,
    \end{equation}
    where $\sigma_z^2$ is the variance of the \ac{AWGN}, and $\sigma_h$ is the Rayleigh fading parameter. 
\end{theorem}
\begin{proof}
    We begin by noting that, given $\mathbf{M}$, $y[t+T]$ and $y[t]$ are zero-mean complex Gaussian \acp{RV} with variances
    \begin{equation}
    \begin{split}
        \sigma_+^2 &= E^2\sigma_h^2\sum_i\left|\langle\vtM_t,\vtX_i\rangle\right|u\left(\langle\vtM_t,\vtX_i\rangle\right) + \sigma_z^2\,\,\,\text{and} \\
        \sigma_-^2 &= E^2\sigma_h^2\sum_i\left|\langle\vtM_t,\vtX_i\rangle\right|u\left(-\langle\vtM_t,\vtX_i\rangle\right) + \sigma_z^2
    \end{split}
    \end{equation}
    respectively. Their powers $|y[t+T]|^2$ and $|y[t]|^2$ are therefore exponential \acp{RV} with rates $\mu_+=1/\sigma_+^2$ and $\mu_-=1/\sigma_-^2$. To compute $b[t]$, we take the difference as
    \begin{equation}
        D \triangleq |y[t+T]|^2 - |y[t]|^2,
    \end{equation}
    which is therefore the difference of two independent exponential \acp{RV}, which in turn is a Laplace \ac{RV}\cite{7019877} with CDF
    \begin{equation}
        F_D(d) = \begin{cases}
            \frac{\sigma_-^2}{\sigma_-^2+\sigma_+^2}e^{d/\sigma_-^2} & \text{when }\,d\leq 0 \\
            1-\frac{\sigma_+^2}{\sigma_-^2+\sigma_+^2}e^{d/\sigma_+^2} & \text{when }\,d> 0.
        \end{cases}
    \end{equation}
    With this CDF we are ready to compute the expected bit as
    \begin{equation}\label{eq:theta}
    \begin{split}
        \mathbb{E}\left[b[t]|\vtM_t\right] &= \Pr\left(D \geq 0\right) - \Pr\left(D < 0\right) = 1 - 2F_D(0) \\
        &= \frac{\sigma_+^2-\sigma_-^2}{\sigma_+^2+\sigma_-^2} = \frac{E^2\sigma_h^2\langle\vtM_t,\vtF\rangle}{E^2\sigma_h^2\sum_{i=1}^n\left|\langle\vtM_t,\vtX_i\rangle\right| + 2\sigma_z^2}.
    \end{split}
    \end{equation}
    Finally, we wish to prove that the positive correlation property \eqref{eq:positive_correlation} holds. From \eqref{eq:th2_theta} and \eqref{eq:theta}, it can be inferred that
    \begin{equation}
        \sign\left(\theta(\langle\vtM_t,\vtF\rangle)\right) = \sign\left(\langle\vtM_t,\vtF\rangle\right).
    \end{equation}
    Since $\theta(\cdot)$ does not change the sign of its argument, it must be true that $\mathbb{E}\left[\theta(\langle \vtM_t,\vtF \rangle) \langle \vtM_t,\vtF \rangle \right] > 0$.
\end{proof}
After collecting the bit-vector $\vtB$, the receiver runs the passive algorithm to reconstruct $\hat{\vtF}$ as
\begin{equation}\label{eq:x_hat2}
    \hat{\vtF} = \begin{cases}
        0, & \text{if } \lVert \frac{1}{T}\mathbf{M}^H\vtB \rVert_{\infty} \leq \gamma, \\
        \frac{1}{\lVert P_{\gamma}\left(\frac{1}{T}\mathbf{M}^H\vtB\right)\rVert_2}P_{\gamma}\left(\frac{1}{T}\mathbf{M}^H\vtB\right) & \text{otherwise},
    \end{cases}
\end{equation}
where $\gamma$ is a hyperparameter, and $P_{\gamma}(\cdot)$ is the soft-thresholding operator from Definition \ref{def:soft_thresh}. Since Theorem \ref{th:mvcs} establishes that
our $\theta$ satisfies the required assumptions, we know that the error bound in Theorem \ref{th:zhang} holds for \ac{MVCS}.

\begin{remark} \label{remark:unit_norm}
    The estimated aggregate from \eqref{eq:x_hat2} is either $\vtZero$ or of unit $\ell_2$-norm. In our system, we assume that $\lVert\vtX_i\rVert_2=1$ for all $i$, but this is not enough information for the receiver to retrieve $\lVert\vtF\rVert_2$. The norm of the aggregate vector may be important, depending on the application. For example, if $\vtF$ is the gradient in distributed gradient descent, the norm $\lVert\vtF\rVert_2$ has an implication on step size and the stopping condition.
    
    Without supplementary information about $\vtX_i$, the only thing known at the receiver is that 
    \begin{equation} \label{eq:2norm}
        0 \leq \lVert\vtF\rVert_2 \leq n.
    \end{equation}
\end{remark}

Now we are ready to present the main result of our work.
\begin{theorem}\label{th:mvcs2}
    Consider the communication scheme where the user devices transmit $\vtA_i\in\mathbb{C}^{2T}$ based on \eqref{eq:transmit_IQ} and \eqref{eq:transmit_IQ_2} with the power control scheme from Proposition \ref{prop:power_control}, and that the receiver forms the received bit-vector $\vtB\in\{-1,1\}^T$ according to \eqref{eq:rx_quantize}. Additionally, consider that the measurement matrix $\mathbf{M}$ is constructed such that $T=\mathcal{O}(kn\log(d)/\epsilon^2)$, where $d$ is the dimension of the desired function $\vtF=\sum_i \vtX_i$, $k\triangleq\lVert\vtX_i\rVert_0$, and $n$ is the number of devices. Then, it holds that the estimate $\hat{\vtF}$ calculated by \eqref{eq:x_hat2} satisfies
    \begin{equation}\label{eq:mvcs_normalized_bound}
        \left\lVert\hat{\vtF} - \frac{\vtF}{\lVert \vtF \rVert_2}\right\rVert_2 \leq \epsilon
    \end{equation}
    with a probability of at least $1-e^{1-\delta}$.
\begin{proof}
    Assuming that $\lambda$ is positive, i.e. 1-bit measurements are positively correlated with the measurements, Theorem~\ref{th:zhang} yields the upper bound \eqref{eq:error} on the estimation error. Note that, in \eqref{eq:error}, both $\hat{\vtX}$ and $\vtX_*$ are of unit norm\footnote{To be precise, $\hat{\vtX}$ is either of unit norm or it is a zero vector. The $\hat{\vtX}=\vtZero$ case has no impact on \eqref{eq:mvcs_normalized_bound} unless $\epsilon\geq1$. If $\epsilon\geq1$, the problem is uninteresting since it is solved by setting $\hat{\vtX}=\vtZero$.}. Therefore,
    \begin{equation}
        \lVert\hat{\vtX}-\vtX_*\rVert_2^2 = \lVert\hat{\vtX}\rVert_2^2 + \lVert\vtX_*\rVert_2^2 - 2\langle \hat{\vtX}, \vtX_*\rangle = 2 - 2\cos(\phi),
    \end{equation}
    where $\phi$ is the angle between $\hat{\vtX}$ and $\vtX_*$. The bound \eqref{eq:error} can therefore equivalently be expressed as
    \begin{equation}
        \sqrt{2 - 2\cos(\phi)} \leq \epsilon \triangleq \frac{3\gamma}{\lambda}\sqrt{\lVert\vtF\rVert_0}.
    \end{equation}
    In other words, we can interpret Theorem \ref{th:zhang} as stating that it is possible to estimate a target vector using a unit vector $\hat{\vtF}$, with a guarantee on the angle $\phi$ between the two vectors.
    Since Theorem~\ref{th:mvcs} verifies that $\lambda$ is positive for \ac{MVCS}, the same angular guarantee holds, even though $\vtF$ is not of unit norm. Another way to express this is to say that $\epsilon$ is bounded with respect to $\vtF/\lVert\vtF\rVert_2$, as in \eqref{eq:mvcs_normalized_bound}. Finally, the order-wise communication cost $T=\mathcal{O}(kn\log(d)/\epsilon^2)$ is given by \eqref{eq:communication_cost}.
\end{proof}
\end{theorem}


\section{MVCS Applied to Machine Learning}
\label{sec:applications}
In this section, we specify how \ac{MVCS} can be leveraged for the application of distributed machine learning. The priority of \ac{MVCS} is to support distributed \ac{SGD}, and therefore we leverage methods related to gradient sparsification. Specifically, we use the $r$Top-$k$ algorithm \cite{barnes2020rtop} and error accumulation \cite{sattler2019sparse}. However, the model generalizes to \ac{FedAvg} by incorporating $E$ local iterations in each communication round and by communicating model updates rather than gradients. If $E=1$, the model corresponds to distributed \ac{SGD}, and if $E>1$ it corresponds to \ac{FedAvg}.

The considered network setup is the same as in Section~\ref{sec:background}, i.e., $n$ devices are communicating with one \ac{AP} over the W-MAC from Definition \ref{def:W-MAC} with Rayleigh fast-fading and large-scale fading according to \eqref{eq:distance} and \eqref{eq:pathloss}. We maintain the communication and power control schemes proposed in Section \ref{sec:proposal} but add pre- and post-processing steps that frame the communication problem within the \ac{ML} application. 

The goal of the \ac{ML} application is to converge on a machine learning model, parameterized by $\vtW\in\mathbb{R}^{d}$, that minimizes some loss function $l(\cdot)$. To find $\vtW$, there is a nested iterative process that finds increasingly better model parameters until the system either runs out of resources or decides that the model has a sufficiently small loss. In this process, the \ac{AP} has the role of coordinating the devices, who are the real work-horses in the training process. The devices are both responsible for running the optimization algorithms, and they are providing the datasets $\mathcal{D}_i$ that are used to train $\vtW$.

Before training begins, the \ac{AP} decides the structure of the \ac{ML} model, selects the learning rate $\mu$, and randomly initiates the model parameters $\vtW^{(0)}$. Communication round $l$ is initiated by the \ac{AP}, which broadcasts the model parameters $\vtW^{(l)}$ to all $n$ devices. Upon receiving the model parameters, the devices run an internal loop for $E\in\mathbb{N}$ iterations to train the model with learning rate $\mu$ as
\begin{equation}\label{eq:training_iteration}
    \vtW_i^{(l)}[e] = \vtW_i^{(l)}[e-1] - \mu\nabla_i^{(l)}[e-1],
\end{equation}
for all $e\in\{1,...,E\}$, where $\vtW_i^{(l)}[0]\triangleq\vtW^{(l)}$ and $\nabla_i^{(l)}[e]$ is the gradient of the loss function $L(\cdot)$ with respect to $\vtW_i^{(l)}$. The gradient is computed as
\begin{equation}\label{eq:gradient}
    \nabla_i^{(l)}[e] = \frac{1}{K}\sum_{k\in\mathcal{K}}\nabla L\left(\vtW_i^{(l)}[e]; \mathcal{D}_i[k]\right),
\end{equation}
where $\mathcal{D}_i[k]$ refers to the $k$th datasample of $\mathcal{D}_i$, and $\mathcal{K}$ is a set of $K$ randomly chosen unique integers from $\{1,..,|\mathcal{D}_i|\}$, i.e., $K$ is the batch size of the mini-batch gradient descent algorithm. After running $E$ iterations of the training algorithm, device $i$ has access to a local model $\vtW_i^{(l)}[E]$. At that point, it computes the model update as
\begin{equation}
    \Delta\vtW_i^{(l)} = \vtW^{(l)} - \vtW_i^{(l)}[E] + \mathbf{\Delta}_i^{(l-1)},
\end{equation}
where $\mathbf{\Delta}_i^{(l-1)}$ is the accumulated sparsification error of the previous communication round, computed in \eqref{eq:error_accumulate}. For the first communication round, $\mathbf{\Delta}_i^{(-1)}=\vtZero$. Once training is complete, the devices sparsify the model updates with the $r$Top-$k$ algorithm \cite{barnes2020rtop} which does two rounds of sparsification. Firstly, all entries of $\Delta\vtW_i^{(l)}$ are set to zero, except for the $r\leq d$ elements of greatest magnitude, i.e.,
\begin{equation}
    \text{top}_r(\Delta\vtW_i^{(l)}) \triangleq \begin{cases}
        \Delta w_i^{(l)}[j]   & \text{if } j\in\{\pi(1), ..., \pi(r)\}, \\
        0                       & \text{otherwise},
    \end{cases}
\end{equation}
where $j\in\{1,...,d\}$ and $\pi$ is a permutation of $\{1,...,d\}$ such that $|\Delta w_i^{(l)}[\pi(i)]|\geq |\Delta w_i^{(l)}[\pi(i+1)]|$ for all $i\in\{1,...,d-1\}$. Secondly, $k\leq r$ random elements of $\text{top}_r(\Delta\vtW_i^{(l)})$ are kept and the remaining elements are dropped, resulting in a $k$-sparse vector $\text{k-top}_r(\Delta\vtW_i^{(l)})$. To align with the notation in Section \ref{sec:proposal}, we denote this vector as
\begin{equation}\label{eq:k_top_r}
    \vtX_i^{(l)} \triangleq \text{k-top}_r(\Delta\vtW_i^{(l)}).
\end{equation}
The information which is lost during sparsification is stored in the memory of the devices, according to the error accumulation algorithm \cite{sattler2019sparse},
\begin{equation}\label{eq:error_accumulate}
    \mathbf{\Delta}_i^{(l)} = \Delta\vtW_i^{(l)} - \vtX_i^{(l)}.
\end{equation}
This accumulated error $\mathbf{\Delta}_i^{(l)}$ is used in the following communication round to compute $\Delta\vtW_i^{(l+1)}$.

At this point, we have $n$ devices carrying $k$-sparse vectors $\vtX_i^{(l)}$ and we are interested in recreating the sum of the sparsified model updates, i.e., $\vtX_*^{(l)} \triangleq \sum_i\vtX_i^{(l)}$. This aligns with the setup for \ac{MVCS}. The transmitter and receiver basebands from Sections \ref{sec:transmitter} and \ref{sec:receiver} are used to communicate and compute the estimate $\hat{\vtX}^{(l)}$ at the \ac{AP} in $2T$ channel uses. The \ac{AP} applies this estimated update as
\begin{equation}\label{eq:model_update}
    \vtW^{(l+1)} = \vtW^{(l)} - \eta^{(l)}\hat{\vtX}^{(l)},
\end{equation}
where $\eta^{(l)}$ is the step size. This concludes communication round $l$. The process is then repeated for multiple rounds until some stopping condition is met, and we use $L$ to denote the total number of communication rounds. The entire \ac{ML} scheme is summarized in Algorithm \ref{alg:MVCS_FL}.
\begin{algorithm}[t]
\caption{Federated Learning via MVCS}
\label{alg:MVCS_FL}
\begin{algorithmic}[1]
\ps
    \State initialize $\vtW^{(0)}$
    \State initialize $\mathbf{\Delta}_i^{(-1)}\leftarrow\vtZero$
\endps
\For{each round $l\in\{0,1,..,L-1\}$}
    \ps
        \State broadcast $\vtW^{(l)}$ to devices
    \endps
    \device 
        \State $\vtW_i^{(l)}[E] \gets$ Equation \eqref{eq:training_iteration}
        \State $\Delta \vtW_{i}^{(l)} \gets \vtW^{(l)} - \vtW_i^{(l)}[E] + \mathbf{\Delta}_i^{(l-1)}$
        \State $\vtX_i \gets \text{k-top}_r(\Delta\vtW_i^{(l)})$
        \State $\mathbf{\Delta}_i^{(l)} = \Delta\vtW_i^{(l)} - \vtX_i^{(l)}$
        \State $\vtS_i \gets \mathbf{M}\vtX_i$
        \State $a_i[t] \gets \sqrt{p_i|s_t[t]|}u(-s_t[t])$
        \State $a_i[t+T] \gets \sqrt{p_i|s_t[t]|}u(s_t[t])$
    \enddevice
    \For{each $t=1,2,...,2T$}
        \devices
            \State transmit $a_i[t]$ over the W-MAC
        \enddevices
        \ps
            \State $y[t] \gets$ Equation \eqref{eq:rx_samples}
        \endps
    \EndFor
    \ps
        \State $b[t] \gets$ Equation \eqref{eq:rx_quantize} for all $t\in\{1,...,T\}$
        \State $\hat{\vtX} \gets$ Equation \eqref{eq:x_hat2}
        
        \State $\eta^{(l)} \gets \eta_s + (\eta_f -\eta_s)\cdot(1-e^{-\eta_d l})$
        \State $\vtW^{(l+1)} \gets \vtW^{(l)} - \eta^{(l)}\hat{\vtX}^{(l)}$
        
    \endps
\EndFor
\end{algorithmic}
\end{algorithm}

\subsection{Step Size}\label{sec:step_size}
In a traditional \ac{FedAvg} setup, the model updates are communicated separately and then combined in the processor of the receiver. This way, the norm of the aggregated model update is available, which is advantageous for convergence. In particular, as training starts from a randomly initialized model, $\vtW^{(0)}$ is generally far from the optimal point, which results in large updates. As the weights approach the optimum, the model update norms decline. Having this information at the receiver helps to select the step size since it is beneficial to take large steps when the weights are far from the optimum. This information is unavailable with the proposed \ac{MV} scheme since all norm information is lost during communication.

With this in mind, the step size schedule for $\eta^{(l)}$ becomes critical when training with \ac{MV}. In this work, we have chosen to leverage an exponential schedule, which aligns with the norms of our numerical results. Specifically,
\begin{equation}\label{eq:learning_schedule}
    \eta^{(l)} = \eta_s + (\eta_f - \eta_s)\cdot(1-e^{-\eta_d l}),
\end{equation}
where $\eta_s$, $\eta_f$, and $\eta_d$ are the starting step size, final step size, and step size decay, respectively. In our numerical results, we select these hyperparameters based on previous training results. In particular, we train a network with an oracle to fit these three parameters. Then, when we train another network with \ac{MVCS}, the step size follows \eqref{eq:learning_schedule} with the fitted parameters. In practice, the choice of $\eta_s$, $\eta_f$, and $\eta_d$ can be based on prior experience of training other models.

\section{MVCS Applied to Histogram Estimation}\label{sec:histogram_estimation}
In this section, we illustrate how \ac{MVCS} can be applied to exact histogram estimation. In our earlier work \cite{hellstrom2024over}, we considered the problem of estimating the histogram support. With \ac{MVCS}, we extend this work to estimate the full histogram. We consider the set-up where $n$ user devices hold a message $\vtX_i\in\{0, 1\}^d$, which are one-hot vectors, and the non-zero element of $\vtX_i$ indicates the location of device $i$'s item. Given this formulation, the histogram can be obtained by computing $\vtX_* = \sum_i\vtX_i$. The goal of the \ac{AP} is to obtain the histogram in a secure way, i.e. directly estimating $\sum_i\vtX_i$ without obtaining knowledge of any individual $\vtX_i$.


Note that, as mentioned in Remark~\ref{remark:unit_norm}, the \ac{MVCS} algorithm yields a unit-norm estimate $\hat{\vtX}\approx\vtX_*/\lVert\vtX_*\rVert_2$.
To retrieve the full histogram $\vtX_*$, additional post-processing steps are required. Since $\vtX_*$ is a sum of one-hot vectors, 
each element of $\vtX_*/\lVert\vtX_*\rVert_2$ is an integer multiple of $1/\lVert\vtX_*\rVert_2$. Furthermore, since the sum of all elements in $\vtX_*$ is equal to $n$, the sum of the elements in the normalized histogram will be $n/\lVert\vtX_*\rVert_2$. Thus, the full histogram can be recovered from its normalized version by first dividing it by the sum of its elements, and then multiplying it by $n$,
\begin{equation}
    \vtX_*=\frac{\vtX_*}{\lVert\vtX_*\rVert_2} \div \left(\sum_{j=1}^d \frac{x_*[j]}{\lVert\vtX_*\rVert_2}\right)  \times n = \frac{\vtX_*}{\lVert\vtX_*\rVert_2} \div \frac{n}{\lVert \vtX_* \rVert_2} \times n, \label{eq:histogram_recover}
\end{equation}
where $j\in\{1,\hdots,d\}$ represents the index of $\vtX_*$.

If the \ac{MVCS} estimate $\hat{\vtX}$ is close to $\vtX_*/\lVert\vtX_*\rVert_2$, we can approximately recover the full histogram $\vtX_*$ by the process similar to \eqref{eq:histogram_recover}. Since the original histogram has at most $n$ non-zero elements, we will keep the $n$ largest elements and define the resulting vector as $\hat{\vtX}_{th}$.
Then we define the estimated full histogram $\hat{\vtX}_f$ as follows, where we first divide $\hat{\vtX}_{th}$ by the sum of its elements and multiply $n$.
\begin{equation}
    \hat{\vtX}_f=\hat{\vtX}_{th} \frac{n}{\sum_j \hat{x}_{th}[j]}.
\end{equation}

From \eqref{eq:x_hat2}, we observe that $\hat{\vtX}$ has non-integer elements, as do $\hat{\vtX}_{th}$ and $\hat{\vtX}_f$. For $\hat{\vtX}$ close enough to $\vtX_*/\lVert\vtX_*\rVert_2$, we can retrieve the exact histogram by rounding each element of $\hat{\vtX}_f$ to the nearest integer. A sufficient condition for exact histogram estimation by rounding is  
\begin{equation} \label{eq:exact_hist_cond}
    \left\Vert \hat{\vtX}_f - \vtX_* \right\Vert_2^2 = \left\Vert \hat{\vtX}_{th} \frac{n}{\sum_j \hat{x}_{th}[j]} - \vtX_* \right\Vert_2^2 \leq \left( \frac{1}{2} \right)^2
\end{equation}
Under this condition, each element of $\hat{\vtX}_f - \vtX_*$ will have an absolute value less than or equal to $1/2$, which leads the rounded version of $\hat{\vtX}_f$ to be equal to $\vtX_*$. However, note that this condition represents a loose bound on exact histogram estimation and it is possible to recover the exact histogram even when the square error exceeds 1/4. For example, consider the case where $\lVert\vtX_*\rVert_0=n$, $\hat{\vtX}_f$ and $\vtX_*$ share the same support, and 
$\max_j|\hat{x}_f[j] - x_*[j]|=\nu<1/2$. Then, $\lVert \hat{\vtX}_f - \vtX_* \rVert_2^2\leq\nu^2 n$, while still yielding the exact histogram by rounding. The bound $\lVert \hat{\vtX}_f - \vtX_* \rVert_2^2\leq\nu^2 n$ is significantly looser than \eqref{eq:exact_hist_cond} for large $n$.

As mentioned in Theorem~\ref{th:mvcs2}, our \ac{MVCS} protocol can guarantee $\left\Vert \hat{\vtX} - \vtX_*/\Vert \vtX_* \Vert_2 \right\Vert_2 \leq \epsilon$ with high probability. We first bound $\Vert \vtX_* \Vert_2$ using the fact that it's a sum of one-hot vectors, and show that \ac{MVCS} can guarantee $\left\Vert \hat{\vtX}_{th} - \vtX_*/\Vert \vtX_* \Vert_2 \right\Vert_2 \leq \epsilon$ as well for $\epsilon$ that is small enough. 

\begin{lemma} \label{lemma:twonorm_bound}
    \begin{equation}
        \sqrt{n} \leq \lVert\vtX_*\rVert_2 \leq n
    \end{equation}
    \begin{proof}
        For the upper bound,
        \begin{align}
            \Vert \vtX_* \Vert_2 &\leq \Vert \vtX_* \Vert_1 = \left\Vert \sum_{i=1}^{n} \vtX_i \right\Vert_1 \underset{\text{(a)}}{\leq} \sum_{i=1}^{n} \Vert \vtX_i \Vert_1 = n,
        \end{align}
        where (a) is given by the triangle inequality. For the lower bound, let $S$ be a set of indices of non-zero elements of $\vtX_*$.
        Since all elements of $\vtX_*$ are nonnegative and $\vtX_*$ can have at most $n$ non-zero elements, by Cauchy-Schwartz inequality,
        \begin{align}
            \left( \sum_{k \in S} x_* [k] \right)^2 \leq n \sum_{k \in S} x_* [k]^2 = n \Vert \vtX_* \Vert_2^2 \Rightarrow n \leq \Vert \vtX_* \Vert_2^2.
        \end{align}
    \end{proof}
\end{lemma}

\begin{lemma} \label{lemma:truncate}
If $0<\epsilon<\frac{1}{2n}$ and $\left\Vert \hat{\vtX} - \vtX_*/\Vert \vtX_* \Vert_2 \right\Vert_2 \leq \epsilon$, then $\left\Vert \hat{\vtX}_{th} - \vtX_*/\Vert \vtX_* \Vert_2 \right\Vert_2 \leq \epsilon$.
\begin{proof}
We first show that, under the conditions of the Lemma, none of the elements of $\hat{\vtX}$ that lie on the support of $\vtX_*$ are truncated to make $\hat{\vtX}_{th}$. In other words, the non-zero elements of $\hat{\vtX} - \hat{\vtX}_{th}$ do not overlap with the support of $\vtX_*$. \\
Without loss of generality, let first $v \in [1, n]$ elements of $\vtX_* / \Vert \vtX_* \Vert_2$ be non-zero and denote them as $w_i (i=1,2, \dots, v)$. All the other elements in $\vtX_* / \Vert \vtX_* \Vert_2$ will be zero. And we define the first $v$ elements of $\hat{\vtX}$ as $y_i (i=1, 2, \dots, v)$, and the remaining $d-v$ elements as $z_i (i=v+1, v+2, \dots, d)$. \\
By squaring both sides of the condition $\left\Vert \hat{\vtX} - \vtX_*/\Vert \vtX_* \Vert_2 \right\Vert_2 \leq \epsilon$, we obtain the following representation.
\begin{equation} \label{eq:l2_loss_expansion}
    \left\Vert \hat{\vtX} - \vtX_*/\Vert \vtX_* \Vert_2 \right\Vert_2^2 = \sum_{i=1}^v (y_i - w_i)^2 + \sum_{i=v+1}^d z_i^2 \leq \epsilon^2 
\end{equation}
Now assume that there exists at least one non-zero element in $\hat{\vtX} - \hat{\vtX}_{th}$ within indices $1$ to $v$, i.e., at least one of $y_i$'s is not truncated. This means there exists at least one $l \in [v+1, d]$ such that $z_l > \min_{i \in [1, v]} y_i$. We'll find the lower bound of $y_i$ to lower bound $z_l$. \\
From \eqref{eq:l2_loss_expansion}, note that $(y_i - w_i)^2 \leq \epsilon^2, \forall i \in [1, v]$. Since $w_i$ can be lower bounded by $1/n$, as the minimum non-zero element in $\vtX_*$ is 1 and the maximum 2-norm of $\vtX_*$ is $n$, we have
\begin{equation}
    |y_i - w_i| \leq \epsilon \Rightarrow w_i - \epsilon \leq y_i \Rightarrow \frac{1}{n} - \epsilon \leq y_i, \forall i \in [1,v].
\end{equation}
Thus $z_l > \frac{1}{n} - \epsilon$ and $z_l^2 > \frac{1}{n^2} - \frac{2\epsilon}{n} + \epsilon^2 > \epsilon^2$, where the last inequality holds since $0 < \epsilon < \frac{1}{2n}$. However, this contradicts \eqref{eq:l2_loss_expansion}, as it makes the left-hand side greater than $\epsilon^2$. By contradiction, $z_l \leq \min_{i \in [1, v]} y_i, \forall l \in [v+1, d]$. \\
Now let $S$ be a set of indices of $n-v$ largest elements among $z_i$. So the non-zero elements of $\hat{\vtX}_{th}$ will be $y_i, i=1, 2, \dots, v$ and $z_i, i \in S$. Then we can represent $\left\Vert \hat{\vtX}_{th} - \vtX_*/\Vert \vtX_* \Vert_2 \right\Vert_2^2$ as follows.
\begin{align}
    \left\Vert \hat{\vtX}_{th} - \vtX_*/\Vert \vtX_* \Vert_2 \right\Vert_2^2 &= \sum_{i=1}^v (y_i - w_i)^2 + \sum_{i \in S} z_i^2 \\
    &\leq \sum_{i=1}^v (y_i - w_i)^2 + \sum_{i=v+1}^d z_i^2 \\
    &=\left\Vert \hat{\vtX} - \vtX_*/\Vert \vtX_* \Vert_2 \right\Vert_2^2 \leq \epsilon^2
\end{align}
\end{proof}
\end{lemma}

To guarantee the exact histogram estimation, our first step is to express the bound on $\left\Vert \hat{\vtX}_f - \vtX_* \right\Vert_2$ using $\left\Vert \hat{\vtX}_{th} - \vtX_*/\Vert \vtX_* \Vert_2 \right\Vert_2$. Then we can bound this function using $\epsilon$, and ultimately find the condition for $\epsilon$ that ensures \eqref{eq:exact_hist_cond} holds. For a bound of $\epsilon$, we will use the one from Lemma~\ref{lemma:truncate}.

\begin{lemma} \label{lemma:hatx_mag}
    If $0<\epsilon<\frac{1}{2n}$ and $\left\Vert \hat{\vtX}_{th} - \frac{\vtX_*}{\Vert \vtX_* \Vert_2} \right\Vert_2 \leq \epsilon$, then
    \begin{equation} \label{eq:exact_hist_bound}
        \left\Vert \hat{\vtX}_f - \vtX_* \right\Vert_2 \leq \Vert \vtX_* \Vert_2 \epsilon + \frac{\Vert \vtX_* \Vert_2 \sqrt{2n} \epsilon}{\frac{n}{\Vert \vtX_* \Vert_2} - \sqrt{2n}\epsilon}
    \end{equation}
    \begin{proof}
        \begin{align}
            &\left\Vert \hat{\vtX}_f - \vtX_* \right\Vert_2 \\
            &= \left\Vert \hat{\vtX}_{th} \frac{n}{\sum_j \hat{x}_{th}[j]} - \vtX_* \right\Vert_2 \\
            &= \left\Vert \left(\hat{\vtX}_{th}\frac{n}{\sum_j \hat{x}_{th}[j]} - \hat{\vtX}_{th}\lVert\vtX_*\rVert_2\right) + \left(\hat{\vtX}_{th} \lVert\vtX_*\rVert_2 - \vtX_*\right) \right\Vert_2 \\
            &\leq \left\Vert \hat{\vtX}_{th} \left( \frac{n}{\sum_j \hat{x}_{th}[j]} - \Vert \vtX_* \Vert_2 \right) \right\Vert_2 + \big\Vert \Vert \vtX_* \Vert_2\hat{\vtX}_{th} - \vtX_* \big\Vert_2 \label{eq:triangle_ineq},
        \end{align}
        where the inequality holds due to the triangle inequality. Now we bound each term in \eqref{eq:triangle_ineq}. We begin with the second term, utilizing the second condition of the Lemma,
        \begin{equation} \label{eq:alg_guarantee}
            \bigg\Vert \Vert \vtX_* \Vert_2\hat{\vtX}_{th} - \vtX_* \bigg\Vert_2 = \Vert \vtX_* \Vert_2 \left\Vert \hat{\vtX}_{th} - \frac{\vtX_*}{\Vert \vtX_* \Vert_2} \right\Vert_2             \leq \Vert \vtX_* \Vert_2 \epsilon.
        \end{equation}
        Then, for the first term in \eqref{eq:triangle_ineq}, we have
        \begin{align} \label{eq:2norm_abs}
            \left\Vert \hat{\vtX}_{th} \left( \frac{n}{\sum_j \hat{x}_{th}[j]} - \Vert \vtX_* \Vert_2 \right) \right\Vert_2 = \Vert \hat{\vtX}_{th} \Vert_2 \left| \frac{n}{\sum_j \hat{x}_{th}[j]} - \Vert \vtX_* \Vert_2 \right|.
        \end{align}
        We go on to bound $\Vert \hat{\vtX}_{th} \Vert_2$ as
        \begin{equation} \label{eq:2norm_hatx}
            \Vert \hat{\vtX}_{th} \Vert_2 \leq \Vert \hat{\vtX} \Vert_2 = 1,
        \end{equation}
        where the inequality comes from the definition of $\hat{\vtX}_{th}$ and the equality holds due to \eqref{eq:x_hat2}\footnote{\eqref{eq:x_hat2} can also yield $\Vert \hat{\vtX} \Vert_2 = 0$ but that case is excluded due to our assumption of $\epsilon < 1/\sqrt{2n}$.}. 
        
        To bound the factor $\left| \frac{n}{\sum_j \hat{x}_{th}[j]} - \Vert \vtX_* \Vert_2 \right|$ in \eqref{eq:2norm_abs}, we consider the following expression to bound $\sum_j \hat{x}_{th}[j]$ as a function of $\Vert \vtX_* \Vert_2$.
        \begin{align}
            &\left| \sum_j \left( \hat{x}_{th}[j] - \frac{x_*[j]}{\Vert \vtX_* \Vert_2} \right) \right| \leq \left\Vert \hat{\vtX}_{th} - \frac{\vtX_*}{\Vert \vtX_* \Vert_2} \right\Vert_1 \\
            & \underset{(a)}{\leq} \sqrt{2n} \left\Vert \hat{\vtX}_{th} - \frac{\vtX_*}{\Vert \vtX_* \Vert_2} \right\Vert_2  \underset{(b)}{\leq} \sqrt{2n}\epsilon,
        \end{align}
        where $(a)$ holds due to the Cauchy-Schwarz inequality as $\hat{\vtX}_{th} - \frac{\vtX_*}{\Vert \vtX_* \Vert_2}$ will have at most $2n$ non-zero elements, and $(b)$ holds due to the second condition of the Lemma.
        Since $\sum_j x_* [j] = n$, we have
        \begin{equation} \label{eq:bound_sum_hatx}
            \frac{n}{\Vert \vtX_* \Vert_2} - \sqrt{2n}\epsilon \leq \sum_j \hat{x}_{th}[j] \leq \frac{n}{\Vert \vtX_* \Vert_2} + \sqrt{2n}\epsilon.
        \end{equation}
        Note that $\frac{n}{\Vert \vtX_* \Vert_2} - \sqrt{2n}\epsilon > 0$
        since $\frac{n}{\Vert \vtX_* \Vert_2} \geq 1$ from Lemma~\ref{lemma:twonorm_bound} and we assumed $\epsilon$ to be less then $\frac{1}{2n}$. Plugging in \eqref{eq:bound_sum_hatx} to $\frac{n}{\sum_j \hat{x}_{th}[j]} - \Vert \vtX_* \Vert_2$, we get
        \begin{align}
            &-\frac{\Vert \vtX_* \Vert_2 \sqrt{2n} \epsilon}{\frac{n}{\Vert \vtX_* \Vert_2} + \sqrt{2n}\epsilon} \leq \frac{n}{\sum_j \hat{x}_{th}[j]} - \Vert \vtX_* \Vert_2 \leq \frac{\Vert \vtX_* \Vert_2 \sqrt{2n} \epsilon}{\frac{n}{\Vert \vtX_* \Vert_2} - \sqrt{2n}\epsilon} \\
            &\Rightarrow \left| \frac{n}{\sum_j \hat{x}_{th}[j]} - \Vert \vtX_* \Vert_2 \right| \leq \frac{\Vert \vtX_* \Vert_2 \sqrt{2n} \epsilon}{\frac{n}{\Vert \vtX_* \Vert_2} - \sqrt{2n}\epsilon}. \label{eq:abs_coeff}
        \end{align}
        Combining \eqref{eq:triangle_ineq}-\eqref{eq:2norm_hatx} and \eqref{eq:abs_coeff}, we can derive the proposed bound.
        
    \end{proof}
\end{lemma}

\begin{corollary} \label{cor:epsilon}
    Condition \eqref{eq:exact_hist_cond} for exact histogram estimation holds if we set $\epsilon$ in Lemma \ref{lemma:hatx_mag} as
    \begin{equation} \label{eq:theoretical_eps}
    \epsilon = \frac{1}{2\sqrt{2} n \sqrt{n} + 2n + \sqrt{2n}}.
    \end{equation}
\begin{proof}
    Since the bound in Lemma \ref{lemma:hatx_mag} is increasing for $\sqrt{n} \leq \Vert \vtX_* \Vert_2 \leq n$, it is maximum if $\Vert \vtX_* \Vert_2 = n$. By setting $\Vert \vtX_* \Vert_2 = n$, we can derive the following bound.
    \begin{equation} \label{eq:norm_bound_2norm}
            \left\Vert \hat{\vtX}_f - \vtX_* \right\Vert_2 \leq n \epsilon + \frac{\sqrt{2} n \sqrt{n} \epsilon}{1 - \sqrt{2n} \epsilon}.
        \end{equation}
    To induce a simplified form, we use a slightly weaker version.
    \begin{equation}
    \left\Vert \hat{\vtX}_f - \vtX_* \right\Vert_2 \leq n \epsilon + (\epsilon + 1) \frac{\sqrt{2} n \sqrt{n} \epsilon}{1 - \sqrt{2n} \epsilon}.
    \end{equation}
    By setting the right-hand side equal to $1/2$, we can derive the expression for $\epsilon$.
\end{proof}
\end{corollary}
\begin{remark}
    It is possible to directly set the right-hand side of Lemma \ref{lemma:hatx_mag} equal to $1/2$ and derive the expression for $\epsilon$. It will still be $O(1/n\sqrt{n})$, which is the same as the result of Corollary~\ref{cor:epsilon}.
\end{remark}
Thus, if we set $\epsilon$ in Lemma \ref{lemma:hatx_mag} as in \eqref{eq:theoretical_eps}, our \ac{MVCS} protocol can retrieve the histogram exactly with high probability.


Next, we derive the number of measurements $T$ that can guarantee \eqref{eq:exact_hist_cond} with high probability. 
From \eqref{eq:communication_cost}, we have a bound on $T$ which is proportional to the reciprocal of $\lambda^2$, where $\lambda$ represents the correlation between the 1-bit measurements and the original measurement as defined in \eqref{eq:lambda_def_2}. Since we aim to find the lower bound for the number of measurements $T$, we should find the lower bound of $\lambda$. 

To do this, we first state our conjecture based on empirical results that $\lambda$ can be lower bounded by $\lambda_g$. As will be described later in the proof of Lemma~\ref{lemma:lambda_g}, $\lambda_g$ is equivalent to $\lambda$ if $n=1$, allowing us to compute $\lambda_g$ empirically by setting $n=1$. Under various combinations of parameters ($d$, $\overline{P}$, $l_{\min}$, $\sigma_h$, $\sigma_z$), we computed both $\lambda$ (for general $n$) and $\lambda_g$ (for $n=1$) using Monte Carlo simulations and observed that $\lambda$ is consistently lower bounded by $\lambda_g$.

\begin{conjecture}
\begin{equation}
    \lambda = \mathbb{E}\left[\theta(\langle \vtM_t,\vtX_* \rangle) \langle \vtM_t,\vtX_* \rangle \right] \geq \mathbb{E}_g\left[\theta(g)g\right] = \lambda_g,
\end{equation}
    where $g$ is a standard Gaussian \ac{RV}.
\end{conjecture}

Since $\lambda_g$ has a simpler form than $\lambda$ for general $n$, we can derive a closed form expression for $\lambda_g=\mathbb{E}_g\left[\theta(g)g\right]$.

\begin{lemma} \label{lemma:lambda_g}
    $\mathbb{E}_g\left[\theta(g)g\right]$ can be represented as
    \begin{equation} \label{eq:lambda_bound}
        \sqrt{\frac{2}{\pi}}-q+q^2\frac{e^{-q^2/2}\left(\pi \textrm{erfi}\left(\frac{q}{\sqrt{2}}\right) - \textrm{Ei}\left(\frac{q^2}{2}\right)\right)}{\sqrt{2\pi}},
    \end{equation}
    where $q=\frac{2\sigma_z^2}{E^2\sigma_h^2}$, erfi denotes the imaginary error function and Ei denotes the exponential integral.
\begin{proof}
    Note that if $n=1$, $\vtX_*=\vtX_1$. Since $\vtM_t$ is a standard Gaussian random vector and $\lVert\vtX_1\rVert_2 =1$, $\langle\vtM_t,\vtX_*\rangle$ follows a standard Gaussian distribution. Thus $\lambda = \mathbb{E}\left[\theta(\langle\vtM_t,\vtX_*\rangle)\langle\vtM_t,\vtX_*\rangle\right]$ is equivalent to  $\lambda_g = \mathbb{E}_g\left[\theta(g)g\right]$ if $n=1$, and we can calculate $\mathbb{E}_g\left[\theta(g)g\right]$ as follows.
    \begin{align}
        &\mathbb{E}_g\left[\theta(g)g\right]\\
        &=\mathbb{E}\left[\theta(\langle\vtM_t,\vtX_*\rangle)\langle\vtM_t,\vtX_*\rangle\right]\\
        &=\mathbb{E}\left[\frac{E^2\sigma_h^2\langle\vtM_t,\vtX_*\rangle^2}{E^2\sigma_h^2\left|\langle\vtM_t,\vtX_1\rangle\right| + 2\sigma_z^2}\right]\\
        &=\mathbb{E}_g\left[\frac{E^2\sigma_h^2 g^2}{E^2\sigma_h^2\left|g\right| + 2\sigma_z^2}\right]\\
        &=\mathbb{E}_g\left[\left|g\right| - \frac{2\sigma_z^2}{E^2\sigma_h^2}+\left(\frac{2\sigma_z^2}{E^2\sigma_h^2}\right)^2\frac{1}{\left|g\right| + \frac{2\sigma_z^2}{E^2\sigma_h^2}}\right]\\
        &=\sqrt{\frac{2}{\pi}}-\frac{2\sigma_z^2}{E^2\sigma_h^2}+\left(\frac{2\sigma_z^2}{E^2\sigma_h^2}\right)^2\mathbb{E}_g\left[\frac{1}{\left|g\right| + \frac{2\sigma_z^2}{E^2\sigma_h^2}}\right] \label{eq:lambda}
    \end{align}
    By using erfi and Ei, we can calculate the expectation in \eqref{eq:lambda} and induce the expression in the lemma.
\end{proof}
\end{lemma}

By plugging in this lower bound of $\lambda$ \eqref{eq:lambda_bound} and the value of $\epsilon$ \eqref{eq:theoretical_eps} for exact histogram estimation into the bound on $T$ \eqref{eq:communication_cost}, we can calculate the analytical bound of $T$ that can guarantee \eqref{eq:exact_hist_cond} with high probability. As stated in the following theorem.

\begin{theorem} \label{th:T_bound}
    Using our \ac{MVCS} scheme, exact histogram estimation is guaranteed with probability $1-e^{1-\delta}$ if the number of measurements satisfies
    \[T \geq \frac{n(\delta+\log d)}{\epsilon_*^2} \left( \frac{6c}{\lambda_*} \right)^2\]
    where $c$ is some constant, $\epsilon_*$ is given in \eqref{eq:theoretical_eps}, and $\lambda_*$ is given in \eqref{eq:lambda_bound}.
\end{theorem}

Compared to the best digital scheme for histogram estimation \cite{chen2023communication}, that requires $O(n \log d)$ bits per device or $O(n^2 \log d)$ bits in total, our scheme had a bound with more channel uses, specifically $O(n^4 \log d)$. However, we emphasize that in practice the complexity required to retrieve the exact histogram is much lower, as demonstrated in Section~\ref{sec:numerical_results}.

\section{Numerical Results}
\label{sec:numerical_results}
In this section, we give numerical results for the two highlighted applications. For both applications, we assume that the devices are randomly uniformly located over a disk of radius $R$ with the \ac{AP} in its center. The PDF of the distance $r_i$ between device $i$ and the \ac{AP} is
\begin{equation}\label{eq:distance}
    f_{r_i}(r) = \frac{2r}{R^2},\, 0\leq r\leq R.
\end{equation}

\subsection{Machine Learning}
To evaluate our proposal in the \ac{ML} setting, we train fully connected deep neural networks for the 10-digit MNIST classification problem. Due to the high computational load of training multiple devices and solving the compressed sensing problem in each communication round, it is challenging to run simulations for large neural networks. Therefore, the neural network in our results is fairly light on parameters. The input layer of the considered neural network contains $28^2$ neurons, one for each pixel in the $28\times28$ gray-scale input images. There is one hidden layer of 16 neurons, and an output layer of 10 neurons, one for each class. This leads to a total of $d=\text{n.o. weights}+\text{n.o. biases} = (28^2\cdot16 + 16\cdot10) + (16 + 10) = 12,730$ parameters. All weights are $l_2$-regularized with parameter $\lambda$, the activation function of the hidden layer is ReLU and the output layer employs softmax. In all simulations, the network contains $n=10$ devices that run mini-batch \ac{SGD} with $K=512$. The dataset is normalized such that all pixels have values in $[0,1]$ and in total there are $60,000$ images of training data which is split into $n=10$ datasets of $6,000$ images each without overlap using a uniform distribution, i.e., the data distribution is IID.

As for the communication network, we consider $\sigma_z^2=0.1$, $\sigma_h=1$, $\overline{P}=1$, $R=2$, and $\beta=2$. Note that $\text{SNR}\ll\overline{P}/\sigma_z^2=10$, due to path loss and power control.

For the passive reconstruction algorithm, we use $\delta=3$ for all simulations, which implies that the error bound holds with a probability $P \geq 1-e^{1-\delta} \approx 86\%$. To select the constant $c$ (or equivalently $\gamma$), we perform hyperparameter optimization. Note that this parameter determines the sparsity of the estimated vector, where the limit of $c=\infty$ yields $\hat{\vtX}=\vtZero$ and $c=0$ makes $\hat{\vtX}$ completely dense, this can be seen in Definition \ref{def:soft_thresh} where $\gamma$ determines the threshold below which elements of the vector are discarded. To find $c$, we run a separate simulation for $L$ Monte-Carlo iterations according to Algorithm \ref{alg:c_star}, before starting the machine learning simulation. As an example, we provide the output of the simulation for $k=5$, $d=12,730$, $T=d/2$, and $n=10$ in Figure \ref{fig:c_star}. 
\begin{algorithm}[t]
\caption{Selecting $c$ for the Passive 1bCS algorithm with $L$ Monte-Carlo iterations.}
\label{alg:c_star}
\begin{algorithmic}[1]
\init
    \State $\vtE \leftarrow \vtZero$ (error vector)
\endinit
\For{each $j\in\{0,1,..,L-1\}$}
    \For{each $i\in\{1,..,n\}$}
        \State $\vtN \leftarrow k$ random integers $\in\{1,...,d\}$
        \State $\vtX_i(\vtN) \leftarrow \text{randn(k, 1)}$ ($\vtN$ is non-zero indices)
        \State $\vtX_i \leftarrow \vtX_i/\lVert\vtX_i\rVert_2$
        \State $\vtS_i \leftarrow \mathbf{M}\vtX_i$
        \State $l_i \leftarrow $ Equations \eqref{eq:distance} \& \eqref{eq:pathloss}
    \EndFor
    \State $\vtX_* \leftarrow \sum_i\vtX_i$
    \State $\mathbf{M} \leftarrow$ randn$(T, d)$
    \State $\vtY \leftarrow$ Equation \eqref{eq:rx_samples}
    \State $\vtB \leftarrow$ Equation \eqref{eq:rx_quantize}
    \For{each $c\in\{0.01,0.02,..,2.99, 3.00\}$}
        \State $\gamma \leftarrow 2c\sqrt{\frac{\delta + \log(d)}{T}}$
        \State $\hat{\vtX} \leftarrow$ Equation \eqref{eq:x_hat2}
        \State $\vtE(c) \leftarrow \vtE(c) + \lVert\vtX_* - \hat{\vtX}\rVert_2^2$
    \EndFor
\EndFor
\State $c^*\leftarrow \text{argmin}(\vtE)$
\end{algorithmic}
\end{algorithm}
\begin{figure}[t]

\begin{tikzpicture}
\begin{axis}[
        xlabel={$c$},
        ylabel={$\text{average error}$},
        x label style={at={(0.5,-0.3)}},
        y label style={at={(-0.15,0.5)}},
        width=9cm,
        height=3.5cm,
        inner sep=0pt,
        outer sep=0pt,
        tick align=outside,
        tick pos=left,
        xmin=0, xmax=3,
        ymin=0, ymax=2,
        xtick={0.5, 1.0, 1.5, 2.0, 2.5},
        ytick={0.50, 1.00, 1.50},
        legend style={nodes={scale=0.75, transform shape}, at={(0.98,0.98)}},
        ymajorgrids=true,
        xmajorgrids=true,
        grid style=dashed,
        grid=both,
        grid style={line width=.1pt, draw=gray!15},
        major grid style={line width=.2pt,draw=gray!40},
    ]
    \addplot[
        color=antiquefuchsia,
        mark=,
        mark options = {rotate = 180},
        line width=1pt,
        mark size=2pt,
        ]
    table[x=c_list,y=cumulative_error_list]
    {Data/cstar.txt};
    \end{axis}
\end{tikzpicture}

    \caption{Hyperparameter optimization of the $c$-constant in the passive 1-bit compressive sensing algorithm. Each value of c is evaluated for $L=100$ random generations of the data vectors $\vtX_i$, the measurement matrix $\mathbf{M}$, the noise, the fading coefficients, and the device locations. If $c$ is chosen as a large number, the algorithm selects $\hat{\vtX}=\vtZero$, which is why the error curve gets flat as $c>1.5$. For this particular setup, the algorithm returns $c^*=0.40$, which is then used in the machine learning simulation.}
    \label{fig:c_star}
\end{figure}
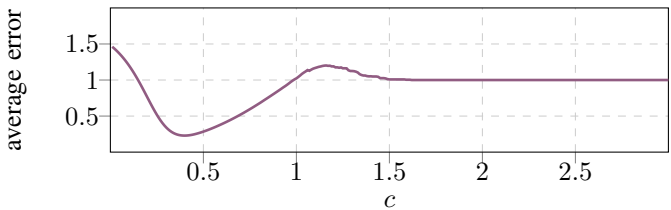

As discussed in Section \ref{sec:step_size}, the norm of the gradient is lost with \ac{MVCS}, which makes the selection of step size difficult. The choice of step size in \ac{MVCS} is an interesting research problem on its own, which can have significant impact on the convergence of the \ac{FedAvg} algorithm. To restrict the scope of our work, we select the step size numerically, i.e., we select $\eta_s$, $\eta_f$, and $\eta_d$ for \eqref{eq:learning_schedule} based on numerical results. In particular, we train the \ac{ML} model under two baselines where communication conditions are perfect. The first is the dense baseline, where $k=d$ in the sparsification step \eqref{eq:k_top_r} and the average model update $1/n\sum_{i}\vtX_i^{(l)}$ is recreated perfectly by an oracle at the receiver. This first baseline corresponds to conditions where communication is distortion-free and resources are unlimited. The second is the sparse baseline, where $r$Top-$k$ sparsification \eqref{eq:k_top_r} is performed with $k<r\ll d$ and the sparse average model update is recreated perfectly by an oracle. This second baseline corresponds to conditions where communication is distortion-free but resources are limited to transmit $k$ elements per device. In Figure \ref{fig:norm}, we illustrate the evolution of the gradient norm, i.e. $\lVert\hat{\vtX}^{(l)}\rVert_2$ from \eqref{eq:model_update}, for $k=5$ and $r=10$. From the result, curve fitting yields $\eta_s=0.15$, $\eta_f=0.02$, $\eta_d=0.01$ in the dense scenario and $\eta_s=0.05$, $\eta_f=0.25$, $\eta_d=0.01$ in the sparse scenario. Generally, the norm of the dense update starts out high and then decreases as the \ac{ML} model approaches a local optimum. The sparse update, however, can grow as the model converges. The main reason for this behavior is the error accumulation step \eqref{eq:error_accumulate}. In every communication round, $k$ elements of the accumulated error $\mathbf{\Delta}_i^{(l)}$ are reset to zero, and the remaining $d-k$ elements change. As long as the gradient is pointing in a similar direction, the magnitude of most of these $d-k$ elements increase every communication round. Since $\mathbf{\Delta}_i^{(l)}$ is added to the transmitted update (see Algorithm \ref{alg:MVCS_FL}), the magnitude of the update tends to increase.

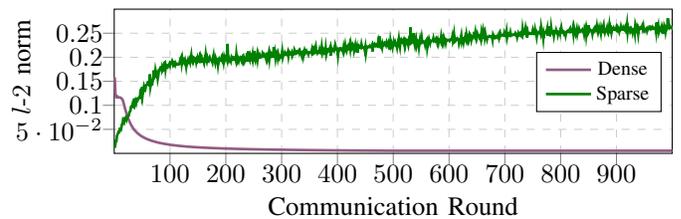
\begin{figure}[t]

\begin{tikzpicture}
\begin{axis}[
        xlabel={Communication Round},
        ylabel={$l$-$2$ norm},
        x label style={at={(0.5,-0.3)}},
        y label style={at={(-0.15,0.6)}},
        width=9cm,
        height=3.5cm,
        inner sep=0pt,
        outer sep=0pt,
        tick align=outside,
        tick pos=left,
        xmin=0, xmax=1000,
        ymin=0, ymax=0.3,
        xtick={100, 200, 300, 400, 500, 600, 700, 800, 900},
        ytick={0.05, 0.10, 0.15, 0.20, 0.25},
        legend style={nodes={scale=0.75, transform shape}, at={(0.98,0.7)}},
        ymajorgrids=true,
        xmajorgrids=true,
        grid style=dashed,
        grid=both,
        grid style={line width=.1pt, draw=gray!15},
        major grid style={line width=.2pt,draw=gray!40},
    ]
    \addplot[
        color=antiquefuchsia,
        mark=,
        mark options = {rotate = 180},
        line width=1pt,
        mark size=2pt,
        ]
    table[x=comm_round,y=norm1]
    {Data/norms.txt};
    \addplot[
        color=cssgreen,
        mark=,
        mark options = {rotate = 180},
        line width=1pt,
        mark size=2pt,
        ]
    table[x=comm_round,y=norm3]
    {Data/norms.txt};
    \legend{Dense, Sparse};
    \end{axis}
\end{tikzpicture}

    \caption{Illustration of gradient norms for dense and sparse distributed gradient descent. Since the norm information is lost when employing \ac{AirComp} with majority vote, the selection of step size is important for the learning algorithm. Unlike dense gradient updates, the sparse version does not necessarily have decreasing gradient norms. The reason for this is the error accumulation step, which appends lost gradient information each time the gradient is sparsified. }
    \label{fig:norm}
\end{figure}

Once $c$, $\eta_s$, $\eta_f$, and $\eta_d$ have been selected, the distributed machine learning simulations begin. For this evaluation, we compare \ac{MVCS} to three baselines. Two baselines are the perfect dense and perfect sparse baselines which we used to compute the step sizes. For these baselines, we consider that the oracles also recreate the norm. Since the norm is available, we choose $\eta_s=\eta_f=1$ for the perfect baselines. The third baseline is the state-of-the-art dense \ac{MV} \ac{AirComp} scheme from \cite{csahin2023distributed}. Essentially, the communication scheme from \ac{MVCS} and \cite{csahin2023distributed} are the same, except that \cite{csahin2023distributed} does not compress the updates with \ac{1bCS}. In each communication round, \cite{csahin2023distributed} allocates $2d$ orthogonal radio resources, and each device communicates in one of the two resources to indicate a positive or negative vote. In our implementation of \cite{csahin2023distributed}, the transmission power is set to $p_i = 2\overline{P}l_{\min}^2/l_i^2$ and $a_i[t]=\sqrt{p_i}$ or $a_i[t]=0$ depending on the vote, which aligns the received powers and satisfies the power constraint. To clarify, \cite{csahin2023distributed} is `dense' in the sense that it uses $2d$ orthogonal radio resources in every communication round. However, we acknowledge that \cite{csahin2023distributed} also can employ sparsification to opt out of transmitting small updates, but not to save communication resources. In our numerical results, we find that the choice of the sparsification threshold for \cite{csahin2023distributed} makes small differences in performance, and decided not to sparsify.

In the first \ac{ML} simulation, we evaluate \ac{MVCS} for distributed \ac{SGD} by setting the number of local iterations $E=1$. In practice, the choice of $E=1$ can be motivated in two ways, either to preserve computational resources at the devices (at the cost of communication resources) or to reduce model drift which naturally occurs for $E>1$ \cite{ozfatura2021fedadc}. In all four schemes, the learning parameters are: $\mu=0.2$, $K=512$, $d=12,730$, and $\lambda = 10^{-3}$. For sparsification, we consider $k=5$ and $r=10$. It is noteworthy that \cite{barnes2020rtop} suggests $r=kn$ but experimentally we found $r=50$ to yield poor results, which is why we reduced it to $r=2k$. This is likely related to the compressed sensing algorithm, since a bigger $r$ implies bigger $\lVert\vtX_*\rVert_0$ for correlated gradients, and the inverse problem of the \ac{1bCS} algorithm is easier when $\lVert\vtX_*\rVert_0$ is smaller, as expressed in \eqref{eq:communication_cost}.

The proposed \ac{MVCS} algorithm is used with $T=d/2$, i.e. $2\times$ compression and $c^*=0.40$ as selected by Algorithm \ref{alg:c_star}. Since our proposal costs half the communication resources of the dense schemes, we allow the sparse schemes to run for $L=1,000$ communication rounds and the dense schemes for $L=500$ rounds. This choice implies that the radio resource consumption for the dense and sparse schemes is the same, equating the total transmission time and consumed transmission energy. For both \ac{AirComp} schemes, we use the step size schedules selected in connection to Figure \ref{fig:norm}, i.e., $\eta_s=0.15$, $\eta_f=0.02$, $\eta_d=0.01$ for \cite{csahin2023distributed} and $\eta_s=0.05$, $\eta_f=0.25$, $\eta_d=0.01$ for \ac{MVCS}. The results of this simulation is presented in Figure \ref{fig:sgd}. The perfect dense scheme converges significantly faster than all alternatives, and even with half as many communication rounds, it outperforms the perfect sparse baseline in terms of classification accuracy. Due to the relatively low compression rate of $2\times$, the proposed \ac{MVCS} has a low estimation error and performs similarly to the perfect sparse scheme. For the first 200 communication rounds, the dense \ac{MV} baseline converges faster than the sparse schemes. However, as the classification accuracy improves, it appears that communication errors slow down convergence and the sparse schemes overtake it. After $L=1,000$ rounds, the proposed \ac{MVCS} scheme has about 2-3\% improved classification accuracy over \cite{csahin2023distributed}.
\begin{figure}[t]

\begin{tikzpicture}
\begin{axis}[
        xlabel={Communication Round},
        ylabel={Classification Accuracy},
        x label style={at={(0.5,-0.15)}},
        y label style={at={(-0.15,0.5)}},
        width=9cm,
        height=5.5cm,
        inner sep=0pt,
        outer sep=0pt,
        tick align=outside,
        tick pos=left,
        xmin=0, xmax=1000,
        ymin=0.8, ymax=0.92,
        xtick={100, 200, 300, 400, 500, 600, 700, 800, 900},
        ytick={0.82, 0.84, 0.86, 0.88, 0.90},
        legend style={nodes={scale=0.75, transform shape}, at={(0.98,0.5)}},
        ymajorgrids=true,
        xmajorgrids=true,
        grid style=dashed,
        grid=both,
        grid style={line width=.1pt, draw=gray!15},
        major grid style={line width=.2pt,draw=gray!40},
    ]
    \addplot[
        color=antiquefuchsia,
        mark=,
        mark options = {rotate = 180},
        line width=1pt,
        mark size=2pt,
        ]
    table[x=comm_round,y=testacc1]
    {Data/sgd_solid_500.txt};
    \addplot[
        color=cssgreen,
        mark=,
        mark options = {rotate = 180},
        line width=1pt,
        mark size=2pt,
        ]
    table[x=comm_round,y=testacc2]
    {Data/sgd_solid_1000.txt};
    \addplot[
        color=carolinablue,
        mark=,
        mark options = {rotate = 180},
        line width=1pt,
        mark size=2pt,
        ]
    table[x=comm_round,y=testacc3]
    {Data/sgd_solid_1000.txt};
    \addplot[
        color=amaranth,
        mark=,
        mark options = {rotate = 180},
        line width=1pt,
        mark size=2pt,
        ]
    table[x=comm_round,y=testacc4]
    {Data/sgd_solid_500.txt};
    \addplot[
        color=antiquefuchsia,
        mark=,
        mark options = {rotate = 180},
        line width=1.5pt,
        mark size=2pt,
        dotted,
        ]
    table[x=comm_round,y=perfect_dense_dot]
    {Data/sgd_dot.txt};
    \addplot[
        color=amaranth,
        mark=,
        mark options = {rotate = 180},
        line width=1.5pt,
        mark size=2pt,
        dotted,
        ]
    table[x=comm_round,y=dense_mv_dot]
    {Data/sgd_dot.txt};
    \legend{Perfect Dense, MVCS, Perfect Sparse, Dense MV \cite{csahin2023distributed}};
    \end{axis}
\end{tikzpicture}

    \caption{Illustration of the classification accuracy on the test dataset for the distributed SGD algorithm. There's four schemes being compared, two that employ dense updates, and two that employ sparse updates. In this example, the sparse MVCS scheme uses half as many radio resources per communication round as the Dense MV scheme. We allow the sparse schemes to run for twice as many communication rounds as the dense schemes to equalize radio resource consumption.}
    \label{fig:sgd}
\end{figure}
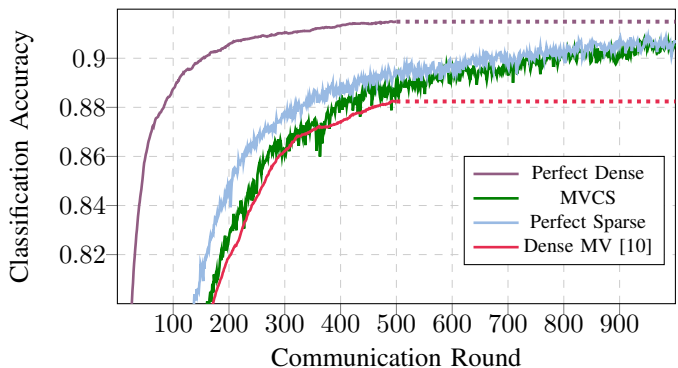

To evaluate the \ac{FedAvg} scenario, we run another simulation with $E=10$ local iterations per communication round. Mostly the same parameters are used as in the \ac{SGD} simulation, with one exception. In particular, we find that in the \ac{FedAvg} setting, the gap between the perfect dense and perfect sparse scheme is too great at $k=5$, and therefore, \ac{MVCS} cannot compete even if its communication performance would be perfect. Instead $k=10$ is used for sparsification and therefore the learning rate $\mu$ is halved to $\mu=0.1$ for the sparse schemes. We keep the learning rate of the dense schemes as $\mu=0.2$, since lowering it would artificially slow them down. Moreover, we rerun the prior optimizations to get $c^*=0.37$ $\eta_s=0.4$, $\eta_f=0.3$, $\eta_d=0.01$ for \ac{MVCS} and $\eta_s=0.3$, $\eta_f=0.02$, $\eta_d=0.01$ for \cite{csahin2023distributed}.

The results of the \ac{FedAvg} simulation is presented in Figure \ref{fig:fl}. As before, the perfect dense scheme converges significantly faster than the others, such that even with half the communication rounds of the perfect sparse scheme, it is superior. Even though $k$ is doubled compared to the previous simulation, the perfect sparse scheme still has a greater gap to the dense counterpart than it does for \ac{SGD}, indicating that sparsification performs worse when $E>1$. Additionally, because the sparsity $k$ has been doubled compared to Figure \ref{fig:sgd}, the reconstruction problem is harder, and we see some gap in performance between the perfect sparse scheme and the proposed \ac{MVCS}. Despite this, the proposed \ac{MVCS} scheme has some advantage over \ac{FedAvg}, but it is less significant than for \ac{SGD}.

We wish to note that while the experiment is set up such that the total radio resource consumption of \ac{MVCS} and dense \ac{MV} is the same, the computational cost of \ac{MVCS} is significantly higher than dense \ac{MV} for two reasons. First, the training algorithm has to run for twice as many iterations when the number of communication rounds is doubled. Second, the compression step at the user devices and the inverse problem at the receiver costs additional computational resources. Therefore, \ac{MVCS} offers improved results when radio resources are limited, but this may not hold if computation is the bottleneck.


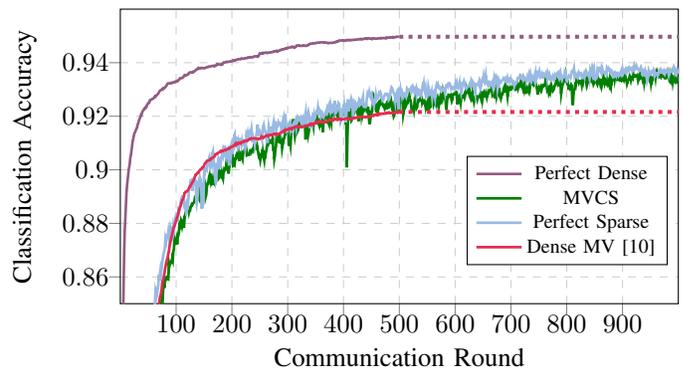
\begin{figure}[t]

\begin{tikzpicture}
\begin{axis}[
        xlabel={Communication Round},
        ylabel={Classification Accuracy},
        x label style={at={(0.5,-0.15)}},
        y label style={at={(-0.15,0.5)}},
        width=9cm,
        height=5.5cm,
        inner sep=0pt,
        outer sep=0pt,
        tick align=outside,
        tick pos=left,
        xmin=0, xmax=1000,
        ymin=0.85, ymax=0.96,
        xtick={100, 200, 300, 400, 500, 600, 700, 800, 900},
        ytick={0.86, 0.88, 0.90, 0.92, 0.94},
        legend style={nodes={scale=0.75, transform shape}, at={(0.98,0.5)}},
        ymajorgrids=true,
        xmajorgrids=true,
        grid style=dashed,
        grid=both,
        grid style={line width=.1pt, draw=gray!15},
        major grid style={line width=.2pt,draw=gray!40},
    ]
    \addplot[
        color=antiquefuchsia,
        mark=,
        mark options = {rotate = 180},
        line width=1pt,
        mark size=2pt,
        ]
    table[x=comm_round,y=testacc1]
    {Data/fl_solid_500.txt};
    \addplot[
        color=cssgreen,
        mark=,
        mark options = {rotate = 180},
        line width=1pt,
        mark size=2pt,
        ]
    table[x=comm_round,y=testacc2]
    {Data/fl_solid_1000.txt};
    \addplot[
        color=carolinablue,
        mark=,
        mark options = {rotate = 180},
        line width=1pt,
        mark size=2pt,
        ]
    table[x=comm_round,y=testacc3]
    {Data/fl_solid_1000.txt};
    \addplot[
        color=amaranth,
        mark=,
        mark options = {rotate = 180},
        line width=1pt,
        mark size=2pt,
        ]
    table[x=comm_round,y=testacc4]
    {Data/fl_solid_500.txt};
    \addplot[
        color=antiquefuchsia,
        mark=,
        mark options = {rotate = 180},
        line width=1.5pt,
        mark size=2pt,
        dotted
        ]
    table[x=comm_round,y=testacc1]
    {Data/fl_dot.txt};
    \addplot[
        color=amaranth,
        mark=,
        mark options = {rotate = 180},
        line width=1.5pt,
        mark size=2pt,
        dotted
        ]
    table[x=comm_round,y=testacc4]
    {Data/fl_dot.txt};
    \legend{Perfect Dense, MVCS, Perfect Sparse, Dense MV \cite{csahin2023distributed}};
    \end{axis}
\end{tikzpicture}

    \caption{Illustration of the classification accuracy on the test dataset for the FL algorithm. Compared to the SGD scenario, the benefit of sparsification is less pronounced for FL.}
    \label{fig:fl}
\end{figure}

\subsection{Histogram Estimation}
For histogram estimation, we use the following parameters: $\sigma_z=0.1$, $\sigma_h=1$, $\overline{P}=1$, $R=2$, and $\beta=2$ for the communication network, and $\delta=3$ for the passive reconstruction algorithm. The constant $c$ is chosen empirically. We conduct 100 Monte-Carlo iterations to calculate the probability of exact histogram estimation. We generate the user item by sampling from a uniform distribution over $d$-dimensional one-hot vectors, unless stated otherwise.

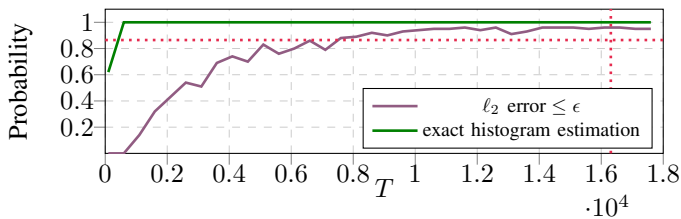
\begin{figure}[t]
\centering

\begin{tikzpicture}
\begin{axis}[
        xlabel={$T$},
        ylabel={Probability},
        x label style={at={(0.5,-0.18)}},
        y label style={at={(-0.13,0.5)}},
        width=9cm,
        height=3.5cm,
        inner sep=0pt,
        outer sep=0pt,
        tick align=outside,
        tick pos=left,
        xmin=0, xmax=18000,
        ymin=0, ymax=1.1,
        ytick={0.2, 0.4, 0.6, 0.8, 1.0},
        legend style={nodes={scale=0.75, transform shape}, at={(0.98,0.45)}},
        ymajorgrids=true,
        xmajorgrids=true,
        grid style=dashed,
        grid=both,
        grid style={line width=.1pt, draw=gray!15},
        major grid style={line width=.2pt,draw=gray!40},
    ]
    \addplot[
        color=antiquefuchsia,
        mark=,
        mark options = {rotate = 180},
        line width=1pt,
        mark size=2pt,
        ]
    table[x=T_list,y=P_error]
    {Data/fig1_d100_n3_c0.15.mat.txt};
    \addplot[
        color=cssgreen,
        mark=,
        mark options = {rotate = 180},
        line width=1pt,
        mark size=2pt,
        ]
    table[x=T_list,y=P_exact]
    {Data/fig1_d100_n3_c0.15.mat.txt};
    \addplot[
        color=amaranth,
        dotted,
        line width=1pt,
        ]
    table[x=x,y=y]
    {Data/prob_thresh.txt};
    \addplot[
        color=amaranth,
        dotted,
        line width=1pt,
        ]
    table[x=x,y=y]
    {Data/T_bound.txt};
    \legend{$\ell_2 \hspace{4pt} \text{error} \leq \epsilon$, exact histogram estimation};
    \end{axis}
\end{tikzpicture}

    \caption{Probability of exact histogram estimation and $\ell_2$ error being bounded by $\epsilon$ in \eqref{eq:theoretical_eps}. For $d=100, n=3,$ and $ c=0.15$. 
    The horizontal dotted line indicates $1-e^{1-\delta} \approx 86\%$ which came from Theorem \ref{th:zhang}, and the vertical dotted line indicates the analytical bound of $T$ from Theorem \ref{th:T_bound} that guarantees $\ell_2$ error to be bounded by $\epsilon$ and thus the exact histogram estimation with the probability above.}
    \vspace{-10pt}
    \label{fig:analytical_T}
\end{figure}

First, we aim to show that the analytical bound of $T$ from Theorem \ref{th:T_bound} can guarantee the exact histogram estimation with high probability. 
For $d=100$ and $n=3$, we count the number of cases where the $\ell_2$ error is less than our threshold $\epsilon$ in \eqref{eq:theoretical_eps}, and also compute the empirical probability of recovering the exact histogram. 
From Figure \ref{fig:analytical_T}, we can observe that the probability of $\ell_2$ error being less than or equal to $\epsilon$ is higher than $1-e^{1-\delta}$ for $T$ greater than the analytical bound, as expected. Additionally, exact histogram recovery is still possible even when the $\ell_2$ error being less than $\epsilon$ is not guaranteed with high probability.




\begin{figure}[t]
\centering

\begin{subfigure}{\textwidth}
\begin{tikzpicture}
\begin{axis}[
        xlabel={$n$},
        ylabel={$\sqrt{T}$},
        x label style={at={(0.5,-0.15)}},
        y label style={at={(-0.1,0.5)}},
        width=9cm,
        height=3.5cm,
        inner sep=0pt,
        outer sep=0pt,
        tick align=outside,
        tick pos=left,
        xmin=0, xmax=35,
        ymin=0, ymax=220,
        legend style={nodes={scale=0.75, transform shape}, at={(0.95,0.3)}},
        ymajorgrids=true,
        xmajorgrids=true,
        grid style=dashed,
        grid=both,
        grid style={line width=.1pt, draw=gray!15},
        major grid style={line width=.2pt,draw=gray!40},
        scatter/classes={
        a={mark=o,draw=blue},
        b={mark=square,draw=red}
      }
    ]
    \addplot[
          scatter,
          only marks,
          scatter src=explicit symbolic
        ]
        table[meta=label] {
        x   y   label
          5	26.45751311	a
10	54.77225575	a
15	89.4427191	a
20	122.4744871	a
25	158.113883	a
30	189.7366596	a
        };
    \addplot[
        color=cssgreen,
        dotted,
        line width=1pt,
        ]
    table[x=x,y=y]
    {Data/trendline.txt};
    \legend{uniform};
    \end{axis}
\end{tikzpicture}
\caption{The user item is uniformly distributed across the all possible one-hot \\
vectors. $c=0.3$.}
\end{subfigure}

\vspace{6 pt}

\begin{subfigure}{\textwidth}
\begin{tikzpicture}
\begin{axis}[
        xlabel={$n$},
        ylabel={$T^{1/3}$},
        x label style={at={(0.5,-0.15)}},
        y label style={at={(-0.1,0.5)}},
        width=9cm,
        height=3.5cm,
        inner sep=0pt,
        outer sep=0pt,
        tick align=outside,
        tick pos=left,
        xmin=0, xmax=35,
        ymin=0, ymax=45,
        legend style={nodes={scale=0.75, transform shape}, at={(0.95,0.3)}},
        ymajorgrids=true,
        xmajorgrids=true,
        grid style=dashed,
        grid=both,
        grid style={line width=.1pt, draw=gray!15},
        major grid style={line width=.2pt,draw=gray!40},
        scatter/classes={
        b={mark=square,draw=red}
      }
    ]
    \addplot[
          scatter,
          only marks,
          scatter src=explicit symbolic
        ]
        table[meta=label] {
        x   y   label
          5	4.791419857	b
10	11.6960709528515	b
15	18.1712059283214	b
20	24.1014226417523	b
25	31.0723250595386	b
30	39.1486764116886	b
        };
    \addplot[
        color=cssgreen,
        dotted,
        line width=1pt,
        ]
    table[x=x,y=y2]
    {Data/trendline.txt};
    \legend{$\Vert \vtX_* \Vert_0 = n/5$};
    \end{axis}
\end{tikzpicture}
\caption{The user items are concentrated in $n/5$ bins. $c=1.2$.}
\end{subfigure}

    \caption{The number of channel uses $T$ required for exact histogram estimation for $d=1000$.}
    \label{fig:empirical_T}
\end{figure}
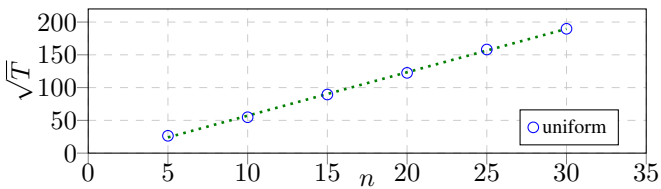
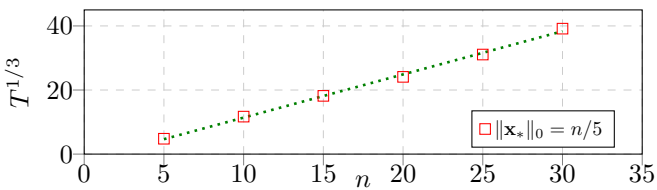

To evaluate the practical number of measurements $T$ required for exact histogram estimation, we perform a grid search for $d=1000$ and different values of $n$. Specifically, we identify the smallest $T$ that can recover the exact histogram with probability at least $1-e^{1-\delta} \approx 86\%$. We consider two different cases for generating user items: (1) they are uniformly distributed across all possible one-hot vectors, or (2) they are concentrated within $n/5$ bins. The results are shown in Figure \ref{fig:empirical_T}. For the case where user items are uniformly distributed, $T$ scales as $\mathcal{O}(n^2)$. In contrast, when the user items are concentrated in $n/5$ bins, $T$ scales as $\mathcal{O}(n^3)$. This difference can be explained using the sparsity of the histogram and the corresponding $\epsilon$. For histograms of uniformly distributed items, the $\ell_2$ norm is typically on the order of $\sqrt{n}$. Substituting this into \eqref{eq:exact_hist_bound} and setting the right-hand side equal to 1/2, we can derive that $\epsilon=\mathcal{O}(1/\sqrt{n})$. Since its sparsity is approximately $n$, total $T$ would be $\mathcal{O}(n^2 \log d)$. On the other hand, when items are concentrated in fewer bins, the histogram is sparser and thus has larger $\ell_2$ norm. Having $\epsilon$ closer to $\mathcal{O}(1/n\sqrt{n})$ and the sparsity of the histogram closer to $\mathcal{O}(1)$, $T$ increases to $\mathcal{O}(n^3 \log d)$.

For uniformly distributed items, our scheme achieves the same order of communication cost as the best digital scheme \cite{chen2023communication}, which requires $\mathcal{O}(n \log d)$ bits per user and $\mathcal{O}(n^2 \log d)$ bits in total. However, \cite{chen2023communication} uses Secure Aggregation to ensure that the AP can recover only the histogram of the items and not the individual items of the users. The secure aggregation protocol requires users to have access to shared keys, which are established ahead of time using a cryptographic protocol with high communication cost. The communication cost of establishing these secure keys is not taken into account in \cite{chen2023communication}.  In contrast, \ac{MVCS} leverages the additive nature of the wireless multiple access channel to ensure that the AP only learns the histogram of the items, eliminating the need for additional cryptographic protocols to preserve privacy.

\begin{figure}

\begin{tikzpicture}
\begin{axis}[
        xlabel={$T$},
        ylabel={Probability},
        x label style={at={(0.5,-0.25)}},
        y label style={at={(-0.1,0.5)}},
        width=9cm,
        height=3.5cm,
        inner sep=0pt,
        outer sep=0pt,
        tick align=outside,
        tick pos=left,
        xmin=0, xmax=5500,
        ymin=0, ymax=1.1,
        ytick={0.2, 0.4, 0.6, 0.8, 1.0},
        legend style={nodes={scale=0.75, transform shape}, at={(0.33,0.85)}},
        ymajorgrids=true,
        xmajorgrids=true,
        grid style=dashed,
        grid=both,
        grid style={line width=.1pt, draw=gray!15},
        major grid style={line width=.2pt,draw=gray!40},
    ]
    \addplot[
        color=antiquefuchsia,
        mark=,
        mark options = {rotate = 180},
        line width=1pt,
        mark size=2pt,
        ]
    table[x=T_list,y=d1e4]
    {Data/compression.txt};
    \addplot[
        color=cssgreen,
        mark=,
        mark options = {rotate = 180},
        line width=1pt,
        mark size=2pt,
        ]
    table[x=T_list,y=d3e4]
    {Data/compression.txt};
    \addplot[
        color=cadmiumorange,
        mark=,
        mark options = {rotate = 180},
        line width=1pt,
        mark size=2pt,
        ]
    table[x=T_list,y=d5e4]
    {Data/compression.txt};
    \addplot[
        color=amaranth,
        dotted,
        line width=1pt,
        ]
    table[x=x,y=y]
    {Data/prob_thresh.txt};
    \legend{$d=1\times10^4$, $d=3\times10^4$, $d=5\times10^4$};
    \end{axis}
\end{tikzpicture}

    \caption{Probability of exact histogram estimation. For $n=10$ and $ c=0.4$ with different $d$ values. The horizontal dotted line indicates $1-e^{1-\delta} \approx 86\%$ from Theorem \ref{th:zhang}.}
    \label{fig:compression}
\end{figure}
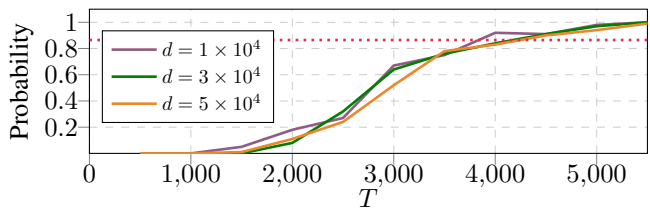

Finally, we illustrate how our scheme can achieve compression for large $d$ values. Since the analytical bound of $T$ from Theorem \ref{th:T_bound} scales with $\log d$, we can expect to achieve higher compression as $d$ increases. Figure \ref{fig:compression} shows the probability of exact histogram estimation for different $d$ values. The $T$ value that achieves the probability of exact histogram estimation proposed from our scheme is approximately 4,000 for $d=1\times10^4, 3\times10^4,$ and $5\times10^4$, and since the number of channel uses is equal to $2T$, it leads to compression ratios of 1.25, 3.75, and 6.25. We couldn't conduct experiments with higher $d$ due to memory issue, but since the bound of $T$ in \eqref{eq:communication_cost} is proportional to $\log d$, the compression ratio will get larger if $d$ is increased further.




\section{Conclusion}
\label{sec:conclusion}
We proposed a non-coherent \ac{AirComp} method for reliable function computation. The proposed approach leverages the state-of-the-art \ac{MV} communication scheme, which uses orthogonal radio resources to indicate the sign of local real-valued vectors. On its own, \ac{MV}-\ac{AirComp} experiences frequent bit errors due to the random phase offsets between the superimposed signals. To correct these errors, we leverage methods from \ac{1bCS}, which serves a similar purpose to a forward error correcting code, offering a receiver-side correction mechanism. Theoretically, we proved that the proposed \ac{MVCS} scheme achieves arbitrarily low errors in $T=\mathcal{O}(kn\log(d))$ channel uses, where $n$ is the number of devices, $k$ is the sparsity of the local vectors, and $d$ is the dimension of the vector. We also provided results for two example applications of \ac{MVCS}: histogram estimation, and distributed \ac{ML}. For histogram estimation, \ac{MVCS} can securely retrieve the histogram in $T=\mathcal{O}(n^4 \log(d))$ channel uses, and in practice, it is achievable with lower complexity. For distributed \ac{ML}, we compared our scheme to a state-of-the-art \ac{AirComp} scheme and found that \ac{MVCS} can reach higher classification accuracy under fixed radio resource constraints. 

Both within histogram estimation and distributed \ac{ML}, the proposed scheme opens up new research directions. In histogram estimation, we suspect that \ac{MVCS} can be improved further since it is not designed to exploit the discrete nature of vectors $\sum_i \vtX_i$. We suspect that sketching \cite{chen2023communication} or quantified group testing \cite{hao1990optimal} could be combined with \ac{AirComp} to design a more efficient method, which remains as future work. In distributed \ac{ML}, the loss of norm information from \ac{1bCS} harms convergence, as described in Section \ref{sec:step_size}. Recent work on \ac{1bCS} with norm estimation could potentially relieve this problem \cite{knudson2016one}.


\bibliographystyle{IEEEtran}
\bibliography{library}

\begin{thebibliography}{10}
\providecommand{\url}[1]{#1}
\csname url@samestyle\endcsname
\providecommand{\newblock}{\relax}
\providecommand{\bibinfo}[2]{#2}
\providecommand{\BIBentrySTDinterwordspacing}{\spaceskip=0pt\relax}
\providecommand{\BIBentryALTinterwordstretchfactor}{4}
\providecommand{\BIBentryALTinterwordspacing}{\spaceskip=\fontdimen2\font plus
\BIBentryALTinterwordstretchfactor\fontdimen3\font minus \fontdimen4\font\relax}
\providecommand{\BIBforeignlanguage}[2]{{%
\expandafter\ifx\csname l@#1\endcsname\relax
\typeout{** WARNING: IEEEtran.bst: No hyphenation pattern has been}%
\typeout{** loaded for the language `#1'. Using the pattern for}%
\typeout{** the default language instead.}%
\else
\language=\csname l@#1\endcsname
\fi
#2}}
\providecommand{\BIBdecl}{\relax}
\BIBdecl

\bibitem{csahin2023survey}
A.~{\c{S}}ahin and R.~Yang, ``{A Survey on Over-the-Air Computation},'' \emph{{IEEE Communications Surveys \& Tutorials (COMST)}}, 2023.

\bibitem{zhang2015stac}
S.~Zhang, X.~Wu, and A.~Ozgur, ``Stac: Simultaneous transmitting and air computing in wireless data center networks,'' in \emph{2015 IEEE/CIC International Conference on Communications in China (ICCC)}, 2015, pp. 1--7.

\bibitem{you2023broadband}
L.~You, X.~Zhao, R.~Cao, Y.~Shao, and L.~Fu, ``{Broadband Digital Over-the-Air Computation for Wireless Federated Edge Learning},'' \emph{IEEE Transactions on Mobile Computing}, 2023.

\bibitem{amiri2020machine}
M.~M. Amiri and D.~G{\"u}nd{\"u}z, ``{Machine Learning at the Wireless Edge: Distributed Stochastic Gradient Descent Over-the-Air},'' \emph{IEEE Transactions on Signal Processing}, vol.~68, pp. 2155--2169, 2020.

\bibitem{fan20221}
X.~Fan, Y.~Wang, Y.~Huo, and Z.~Tian, ``{1-bit Compressive Sensing for Efficient Federated Learning Over the Air},'' \emph{IEEE Transactions on Wireless Communications}, vol.~22, no.~3, pp. 2139--2155, 2022.

\bibitem{edin2024over}
A.~Edin and Z.~Chen, ``{Over-the-Air Federated Learning with Compressed Sensing: Is Sparsification Necessary?}'' in \emph{2024 IEEE International Conference on Machine Learning for Communication and Networking (ICMLCN)}.\hskip 1em plus 0.5em minus 0.4em\relax IEEE, 2024, pp. 287--292.

\bibitem{candes2008introduction}
E.~J. Cand{\`e}s and M.~B. Wakin, ``{An Introduction to Compressive Sampling},'' \emph{IEEE Signal Processing Magazine}, vol.~25, no.~2, pp. 21--30, 2008.

\bibitem{amiri2020federated}
M.~M. Amiri and D.~G{\"u}nd{\"u}z, ``{Federated Learning over Wireless Fading Channels},'' \emph{IEEE Transactions on Wireless Communications}, vol.~19, no.~5, pp. 3546--3557, 2020.

\bibitem{goldenbaum2013robust}
M.~Goldenbaum and S.~Stanczak, ``{Robust Analog Function Computation via Wireless Multiple-Access Channels},'' \emph{{IEEE Transactions on Communications (TCOM)}}, vol.~61, no.~9, pp. 3863--3877, 2013.

\bibitem{csahin2023distributed}
A.~{\c{S}}ahin, ``{Distributed Learning over a Wireless Network with Non-Coherent Majority Vote Computation},'' \emph{IEEE Transactions on Wireless Communications}, 2023.

\bibitem{boufounos2008}
P.~T. Boufounos and R.~G. Baraniuk, ``{1-bit Compressive Sensing},'' in \emph{42nd Annual Conference on Information Sciences and Systems}.\hskip 1em plus 0.5em minus 0.4em\relax IEEE, 2008, pp. 16--21.

\bibitem{candes2006stable}
E.~J. Candes, J.~K. Romberg, and T.~Tao, ``{Stable Signal Recovery from Incomplete and Inaccurate Measurements},'' \emph{Communications on Pure and Applied Mathematics: A Journal Issued by the Courant Institute of Mathematical Sciences}, vol.~59, no.~8, pp. 1207--1223, 2006.

\bibitem{plan2012robust}
Y.~Plan and R.~Vershynin, ``{Robust 1-bit Compressed Sensing and Sparse Logistic Regression: A Convex Programming Approach},'' \emph{IEEE Transactions on Information Theory}, vol.~59, no.~1, pp. 482--494, 2012.

\bibitem{boufounos2009greedy}
P.~T. Boufounos, ``{Greedy Sparse Signal Reconstruction from Sign Measurements},'' in \emph{Conference Record of the Forty-Third Asilomar Conference on Signals, Systems and Computers}.\hskip 1em plus 0.5em minus 0.4em\relax IEEE, 2009, pp. 1305--1309.

\bibitem{plan2013one}
Y.~Plan and R.~Vershynin, ``{One-bit Compressed Sensing by Linear Programming},'' \emph{Communications on pure and Applied Mathematics}, vol.~66, no.~8, pp. 1275--1297, 2013.

\bibitem{zhang2014efficient}
L.~Zhang, J.~Yi, and R.~Jin, ``{Efficient Algorithms for Robust one-bit Compressive Sensing},'' in \emph{International Conference on Machine Learning}.\hskip 1em plus 0.5em minus 0.4em\relax PMLR, 2014, pp. 820--828.

\bibitem{xiao20191}
P.~Xiao, B.~Liao, X.~Huang, and Z.~Quan, ``{1-bit Compressive Sensing with an Improved Algorithm Based on Fixed-Point Continuation},'' \emph{Signal processing}, vol. 154, pp. 168--173, 2019.

\bibitem{xiao2019one}
P.~Xiao, B.~Liao, and J.~Li, ``{One-bit Compressive Sensing via Schur-Concave Function Minimization},'' \emph{IEEE Transactions on Signal Processing}, vol.~67, no.~16, pp. 4139--4151, 2019.

\bibitem{genzel2020robust}
M.~Genzel and A.~Stollenwerk, ``{Robust 1-bit Compressed Sensing via Hinge Loss Minimization},'' \emph{Information and Inference: A Journal of the IMA}, vol.~9, no.~2, pp. 361--422, 2020.

\bibitem{matsumoto2024robust}
N.~Matsumoto and A.~Mazumdar, ``{Robust 1-bit Compressed Sensing with Iterative Hard Thresholding},'' in \emph{Proceedings of the Annual ACM-SIAM Symposium on Discrete Algorithms (SODA)}.\hskip 1em plus 0.5em minus 0.4em\relax SIAM, 2024, pp. 2941--2979.

\bibitem{abari2016over}
O.~Abari, H.~Rahul, and D.~Katabi, ``{Over-the-Air Function Computation in Sensor Networks},'' \emph{arXiv preprint arXiv:1612.02307}, 2016.

\bibitem{ieee80211ax}
``{IEEE Standard for Information Technology--Telecommunications and Information Exchange between Systems Local and Metropolitan Area Networks--Specific Requirements Part 11: Wireless LAN Medium Access Control (MAC) and Physical Layer (PHY) Specifications Amendment 1: Enhancements for High-Efficiency WLAN},'' \emph{{IEEE Std 802.11ax-2021 (Amendment to IEEE Std 802.11-2020)}}, pp. 1--767, 2021.

\bibitem{5gspecification}
3GPP, ``{3GPP} technical specification {TS} 38.104, {NR}; base station radio transmission and reception,'' \url{https://portal.3gpp.org/desktopmodules/Specifications/SpecificationDetails.aspx?specificationId=3202}.

\bibitem{csahin2023over}
A.~{\c{S}}ahin, ``{Over-the-Air Computation Based on Balanced Number Systems for Federated Edge Learning},'' \emph{IEEE Transactions on Wireless Communications}, 2023.

\bibitem{mohammadi2019collaborative}
M.~Mohammadi~Amiri, T.~Duman, D.~Gunduz \emph{et~al.}, ``{Collaborative Machine Learning at the Wireless Edge with Blind Transmitters},'' in \emph{GlobalSIP 2019-7th IEEE Global Conference on Signal and Information Processing, Proceedings}.\hskip 1em plus 0.5em minus 0.4em\relax Institute of Electrical and Electronics Engineers Inc., 2019, pp. 1--5.

\bibitem{sery2020analog}
T.~Sery and K.~Cohen, ``{On Analog Gradient Descent Learning Over Multiple Access Fading Channels},'' \emph{IEEE Transactions on Signal Processing}, vol.~68, pp. 2897--2911, 2020.

\bibitem{7019877}
L.~M. Leemis, ``{Computational Probability Applications},'' in \emph{Proceedings of the Winter Simulation Conference 2014}, 2014, pp. 51--65.

\bibitem{barnes2020rtop}
L.~P. Barnes, H.~A. Inan, B.~Isik, and A.~{\"O}zg{\"u}r, ``{rTop-k: A Statistical Estimation Approach to Distributed SGD},'' \emph{IEEE Journal on Selected Areas in Information Theory}, vol.~1, no.~3, pp. 897--907, 2020.

\bibitem{sattler2019sparse}
F.~Sattler, S.~Wiedemann, K.-R. M{\"u}ller, and W.~Samek, ``{Sparse Binary Compression: Towards Distributed Deep Learning with Minimal Communication},'' in \emph{International Joint Conference on Neural Networks (IJCNN)}.\hskip 1em plus 0.5em minus 0.4em\relax IEEE, 2019, pp. 1--8.

\bibitem{hellstrom2024over}
H.~Hellström, J.~Jeong, W.-N. Chen, A.~Özgür, V.~Fodor, and C.~Fischione, ``Over-the-air histogram estimation,'' in \emph{ICC 2024 - IEEE International Conference on Communications}, 2024, pp. 4717--4722.

\bibitem{chen2023communication}
W.-N. Chen, A.~Ozgur, G.~Cormode, and A.~Bharadwaj, ``The communication cost of security and privacy in federated frequency estimation,'' in \emph{Proceedings of The 26th International Conference on Artificial Intelligence and Statistics}, ser. Proceedings of Machine Learning Research, F.~Ruiz, J.~Dy, and J.-W. van~de Meent, Eds., vol. 206.\hskip 1em plus 0.5em minus 0.4em\relax PMLR, 25--27 Apr 2023, pp. 4247--4274.

\bibitem{ozfatura2021fedadc}
E.~Ozfatura, K.~Ozfatura, and D.~G{\"u}nd{\"u}z, ``{FedADC: Accelerated Federated Learning with Drift Control},'' in \emph{IEEE International Symposium on Information Theory (ISIT)}.\hskip 1em plus 0.5em minus 0.4em\relax IEEE, 2021, pp. 467--472.

\bibitem{hao1990optimal}
F.~H. Hao, ``{The Optimal Procedures for Quantitative Group Testing},'' \emph{Discrete Applied Mathematics}, vol.~26, no.~1, pp. 79--86, 1990.

\bibitem{knudson2016one}
K.~Knudson, R.~Saab, and R.~Ward, ``{One-bit Compressive Sensing with Norm Estimation},'' \emph{IEEE Transactions on Information Theory}, vol.~62, no.~5, pp. 2748--2758, 2016.

\end{thebibliography}

\end{document}